\documentclass[11pt]{article}

\usepackage[sc]{mathpazo}

\usepackage[T1]{fontenc}
\usepackage[margin=1in]{geometry}
\usepackage{mathtools}
\usepackage{color}
\usepackage{proof}
\usepackage{caption}
\usepackage{bbm}
\usepackage{amssymb}
\usepackage{microtype}
\usepackage{xcolor}
\usepackage{enumitem}
\usepackage{amsthm}
\usepackage{thmtools,thm-restate}
\usepackage{mdframed}

\usepackage{amsmath}
\usepackage[alphabetic]{amsrefs}
\usepackage{url}

\usepackage[colorlinks,citecolor=blue]{hyperref}
\newtheorem{theorem}{Theorem}[section]
\newtheorem{lemma}[theorem]{Lemma}
\newtheorem{definition}[theorem]{Definition}

\newtheorem{remark}[theorem]{Remark}

\newmdtheoremenv[backgroundcolor=gray!10,
                 linewidth=0pt,
                 innerleftmargin=4pt,
                 innerrightmargin=4pt,
                 innertopmargin=-2pt,
                 innerbottommargin=4pt,
                 splitbottomskip=4pt]{protocol}[prot]{Protocol}

\newlist{proto}{description}{1}
\setlist[proto]{align=right,labelindent=1em,labelwidth=1.5cm,leftmargin=!,itemsep=0pt}

\newcounter{prob}
\newtheorem{problem}[prob]{Problem}

\newcommand{\eq}[1]{\hyperref[eq:#1]{(\ref*{eq:#1})}}
\renewcommand{\sec}[1]{\hyperref[sec:#1]{Section~\ref*{sec:#1}}}
\newcommand{\thm}[1]{\hyperref[thm:#1]{Theorem~\ref*{thm:#1}}}
\newcommand{\lem}[1]{\hyperref[lem:#1]{Lemma~\ref*{lem:#1}}}
\newcommand{\cor}[1]{\hyperref[cor:#1]{Corollary~\ref*{cor:#1}}}
\newcommand{\itm}[1]{\hyperref[itm:#1]{\ref*{itm:#1}}}
\newcommand{\app}[1]{\hyperref[app:#1]{Appendix~\ref*{app:#1}}}
\newcommand{\dfn}[1]{\hyperref[dfn:#1]{Definition~\ref*{dfn:#1}}}
\newcommand{\fig}[1]{\hyperref[fig:#1]{Figure~\ref*{fig:#1}}}
\newcommand{\clm}[1]{\hyperref[clm:#1]{Claim~\ref*{clm:#1}}}
\newcommand{\alg}[1]{\hyperref[alg:#1]{Algorithm~\ref*{alg:#1}}}
\newcommand{\stp}[1]{\hyperref[stp:#1]{Step~\ref*{stp:#1}}}
\newcommand{\asm}[1]{\hyperref[asm:#1]{Assumption~\ref*{asm:#1}}}
\newcommand{\prot}[1]{\hyperref[prot:#1]{Protocol~\ref*{prot:#1}}}
\newcommand{\prob}[1]{\hyperref[prob:#1]{Problem~\ref*{prob:#1}}}
\newcommand{\rmk}[1]{\hyperref[rmk:#1]{Remark~\ref*{rmk:#1}}}

\newcommand{\expref}[2]{\texorpdfstring{\hyperref[#2]{#1~\ref{#2}}}{#1~\ref{#2}}}

\newcommand{\bra}[1]{\langle #1 \vert}
\newcommand{\ket}[1]{\vert #1 \rangle}
\newcommand{\proj}[1]{\vert #1\rangle\!\langle #1\vert}

\DeclareMathOperator{\tr}{tr}

\let\originalleft\left
\let\originalright\right
\renewcommand{\left}{\mathopen{}\mathclose\bgroup\originalleft}
\renewcommand{\right}{\aftergroup\egroup\originalright}

\newcommand{\A}[0]{\mathcal{A}}
\newcommand{\B}[0]{\mathcal{B}}
\newcommand{\C}[0]{\mathcal{C}}

\newcommand{\F}[0]{\mathcal{F}}
\newcommand{\G}[0]{\mathcal{G}}
\renewcommand{\H}[0]{\mathcal{H}}
\newcommand{\K}[0]{\mathcal{K}}

\newcommand{\R}[0]{\mathcal{R}}

\newcommand{\W}[0]{\mathcal{W}}
\newcommand{\X}[0]{\mathcal{X}}
\newcommand{\Y}[0]{\mathcal{Y}}

\newcommand{\zx}[0]{\textsc{zx}}
\DeclareMathOperator*{\Exp}{\mathbb{E}}
\DeclareMathOperator{\poly}{poly}
\DeclareMathOperator{\negl}{negl}

\newcommand{\yes}{\mathrm{yes}}
\newcommand{\no}{\mathrm{no}}

\renewcommand{\Re}{\mathop{\mathrm{Re}}}

\newcommand{\NP}[0]{\ensuremath{\mathsf{NP}}\xspace}
\newcommand{\QMA}[0]{\mathsf{QMA}}

\usepackage{xspace}
\newcommand{\QPIA}[0]{\ensuremath{\mathrm{QPIA}}\xspace}
\newcommand{\FS}{\ensuremath{\mathsf{FS}}\xspace}
\newcommand{\BQP}[0]{\ensuremath{\mathsf{BQP}}\xspace}
\newcommand{\MA}[0]{\ensuremath{\mathsf{MA}}\xspace}
\newcommand{\LWE}[0]{\ensuremath{\mathsf{LWE}}\xspace}

\newcommand{\1}[0]{\mathbbm{1}}
\newcommand{\supp}[0]{\mathrm{supp\,}}

\newcommand{\ZZ}[0]{\mathbb{Z}}
\newcommand{\LL}[0]{\mathcal{L}}

\newcommand{\Gen}[0]{\mathsf{Gen}}

\newcommand{\Chk}[0]{\mathsf{Chk}}
\newcommand{\Inv}[0]{\mathsf{Inv}}
\newcommand{\Samp}[0]{\mathsf{Samp}}
\newcommand{\FHE}[0]{\mathsf{FHE}}
\newcommand{\Enc}[0]{\mathsf{Enc}}
\newcommand{\Dec}[0]{\mathsf{Dec}}
\newcommand{\Eval}[0]{\mathsf{Eval}}
\newcommand{\Com}[0]{\mathsf{Com}}
\newcommand{\gen}[0]{\mathsf{gen}}
\newcommand{\commit}[0]{\mathsf{commit}}
\newcommand{\verify}[0]{\mathsf{verify}}
\newcommand{\Sim}[0]{\mathsf{Sim}}
\newcommand{\Ext}[0]{\mathsf{Ext}}
\newcommand{\aP}[0]{\mathsf{P}}
\newcommand{\aV}[0]{\mathsf{V}}
\newcommand{\aS}[0]{\mathsf{S}}
\newcommand{\NIZK}[0]{\mathsf{NIZK}}
\newcommand{\Setup}[0]{\mathsf{Setup}}
\newcommand{\crs}[0]{\mathsf{crs}}
\newcommand{\term}[0]{\mathsf{term}}
\newcommand{\verdict}[0]{\mathsf{verdict}}
\newcommand{\witness}[0]{\tau}
\newcommand{\setup}[0]{\mathsf{setup}}
\newcommand{\st}[0]{\mathsf{st}}

\newcommand{\test}{\ensuremath{\mathfrak{t}}\xspace}
\newcommand{\hada}{\ensuremath{\mathfrak{h}}\xspace}

\newcommand{\genprot}[0]{generalized $\Sigma$-protocol\xspace}
\newcommand{\genprots}[0]{generalized $\Sigma$-protocols\xspace}

\newcommand{\prover}{\ensuremath{\mathcal P}\xspace}
\newcommand{\verifier}{\ensuremath{\mathcal V}\xspace}
\newcommand{\simulator}{\ensuremath{\mathcal S}\xspace}

\newcommand{\PPT}{PPT\xspace}
\newcommand{\QPT}{QPT\xspace}

\newcommand{\bit}{\{0,1\}}

\newcommand{\hybrids}{
  \begin{enumerate}
  \item[$H_0$:] \prot{interactive-attempt}.
  \item[$H_1$:] Same as $H_0$ except that in the setup phase, $\prover$ receives uniform one-time pad keys $\beta,\gamma$ and randomness $r_1$ for the commitment and $\verifier$ receives the commitment $\xi=\commit(\beta,\gamma, r_1)$.
    Moreover, $\prover$ reveals $\beta,\gamma$ and $r_1$ in Round $\prover_2$.
    $\verifier$ accepts if $\verdict'$ returns 1 and $\xi=\commit(\beta,\gamma, r_1)$.
  \item[$H_2$:] Same as $H_1$ except that $\prover$ also sends the commitment to $u$ with message $\chi$, and the randomness $r_2$.
    $\verifier$ accepts if $\verdict'$ returns 1, $\xi=\commit(\beta,\gamma,r_1)$, and $\chi=\commit(u,r_2)$.
  \item[$H_3$:] Same as $H_2$ except that in the setup phase, both $\verifier$ and $\prover$ also receive $csk=\FHE.\Enc_{hpk}(sk)$ and $cs=\FHE.\Enc_{hpk}(s)$.
  \item[$H_4$:] Same as $H_3$ except that in the setup phase both $\verifier$ and $\prover$ get $\crs$ and in Round $\prover_3$, $\prover$ sends $ce$ and in Verdict, $\verifier$ accepts if $ce$ is a valid ciphertext and $\NIZK.\aV(\crs,x,\Dec_{hsk}(ce))=1$.
Notice that this is \prot{interactive}.
  \end{enumerate}
}

\begin{document}

\title{Non-interactive classical verification\\ of quantum computation}
\author{Gorjan Alagic$^{1,2,3}$\quad Andrew M. Childs$^{1,3,4}$ \quad  Alex B. Grilo$^{5,6}$ \quad Shih-Han Hung$^{1,3,4}$ \\[10pt]
\small{$^1$Joint Center for Quantum Information and Computer Science, University of Maryland} \\
\small{$^2$National Institute of Standards and Technology, Gaithersburg, Maryland} \\
\small{$^3$Institute for Advanced Computer Studies, University of Maryland} \\
\small{$^4$Department of Computer Science, University of Maryland} \\
\small{$^5$CWI, Amsterdam} \\
\small{$^6$QuSoft, Amsterdam} \\ [10pt]
}

\date{}

\maketitle

\begin{abstract}
In a recent breakthrough, Mahadev constructed an interactive protocol that enables a purely classical party to delegate any quantum computation to an untrusted quantum prover. In this work, we show that this same task can in fact be performed \emph{non-interactively} and in \emph{zero-knowledge.}

  Our protocols result from a sequence of significant improvements to the original four-message protocol of Mahadev. We begin by making the first message instance-independent and moving it to an offline setup phase. We then establish a parallel repetition theorem for the resulting three-message protocol, with an asymptotically optimal rate. This, in turn, enables an application of the Fiat-Shamir heuristic, eliminating the second message and giving a non-interactive protocol. Finally, we employ classical non-interactive zero-knowledge (NIZK) arguments and classical fully homomorphic encryption (FHE) to give a zero-knowledge variant of this construction. This yields the first purely classical NIZK argument system for $\QMA$, a quantum analogue of $\NP$.

We establish the security of our protocols under standard assumptions in quantum-secure cryptography. Specifically, our protocols are secure in the Quantum Random Oracle Model, under the assumption that Learning with Errors is quantumly hard. The NIZK construction also requires circuit-private FHE.
\end{abstract}

\section{Introduction}

Quantum computing devices are expected to solve problems that are infeasible for classical computers. However, as significant progress is made toward constructing quantum computers, it is challenging to verify that they work correctly. This becomes particularly difficult when devices reach scales that rule out direct classical simulation.

This problem has been considered in various models, such as with multiple entangled quantum provers~\cite{RUV13,McK16,GKW15,HPDF15,FH15,NV17,CGJV19,Grilo19} or with verifiers who have limited quantum resources~\cite{BFK09,Bro18,MF16,ABOEM17}. Such solutions are not ideal since they require assumptions about the ability of the provers to communicate or require the verifier to have some quantum abilities.

In a major breakthrough, Mahadev recently described the first secure protocol enabling a purely classical verifier to certify the quantum computations of a single untrusted quantum prover~\cite{Mah18}. The Mahadev protocol uses a quantum-secure cryptographic assumption to give the classical verifier leverage over the quantum prover. Specifically, the protocol is sound under the assumption that the Learning with Errors (LWE) problem does not admit a polynomial-time quantum algorithm. This assumption is widely accepted, and underlies some of the most promising candidates for quantum-secure cryptography~\cite{NIST19}.

\paragraph{The Mahadev protocol.}

Mahadev's result settled a major open question concerning the power of \emph{quantum-prover interactive arguments} (QPIAs). In a QPIA, two computationally-bounded parties (a quantum prover $\prover$ and a classical verifier $\verifier$) interact with the goal of solving a decision problem. Mahadev's result showed that there is a four-round\footnote{We take one round to mean a single one-way message from the prover to the verifier, or vice-versa. The Mahadev protocol involves four such messages.} QPIA for $\BQP$ with negligible completeness error and constant soundness error $\delta \approx 3/4$. The goal of the protocol is for the verifier to decide whether an input Hamiltonian $H$ from a certain class (which is $\BQP$-complete) has a ground state energy that is low (YES) or high (NO). 

The protocol has a high-level structure analogous to classical $\Sigma$-protocols~\cite{Dam02}:
\begin{enumerate}
\item \verifier generates a private-public key pair $(pk,sk)$ and sends $pk$ to \prover;
\item \prover prepares the ground state of $H$ and then coherently evaluates a certain classical function $f_{pk}$. This yields a state of the form
\begin{align}\label{eq:mah-sketch}
\sum_x \alpha_x \ket{x}_X \ket{f_{pk}(x)}_Y\,,
\end{align}
where the ground state is in a subregister of register $X$. \prover measures the output register $Y$ and sends the result $y$ to \verifier. Note that \prover now holds a superposition over the preimages of $y$.
\item \verifier replies with a uniformly random \emph{challenge} bit $c \in \{0,1\}$.
\item If $c=0$ (``test round''), \prover measures the $X$ register in the computational basis and sends the outcome. If $c=1$ (``Hadamard round''), \prover measures $X$ in the Hadamard basis and sends the outcome.
\end{enumerate}
After the four message rounds above are completed, the verifier uses their knowledge of $H$ and the secret key $sk$ to either accept or reject the instance $H$.

\paragraph{Our results.}

In this work, we show that the Mahadev protocol can be transformed into protocols with significantly more favorable parameters, and with additional properties of interest.  Specifically, we show how to build non-interactive protocols (with setup) for the same task, with negligible completeness and soundness errors. One of our protocols enables a verifier to publish a single public ``setup'' string and then receive arbitrarily many proofs from different provers, each for a different instance. We also construct a non-interactive protocol that satisfies the zero-knowledge property~\cite{BFM88}.

In principle, one could ask for a slightly less interactive protocol: one where the prover and the verifier both receive the instance from some third party, and then the prover simply sends a proof to the verifier, with no setup required. While we cannot rule such a protocol out, constructing it seems like a major challenge (and may even be impossible). In such a setting, the proof must be independent of the secret randomness of the verifier, making it difficult to apply the ``cryptographic leash'' technique of Mahadev. On the other hand, without cryptographic assumptions, such a protocol would result in the unlikely inclusion \BQP $\subseteq$ \MA~\cite{Aar10}.

All of our results are conditioned on the hardness of the \LWE{} problem for quantum computers; we call this {\em the \LWE assumption}. This assumption is inherited from the Mahadev protocol. For the zero-knowledge protocol, we also require fully-homomorphic encryption (FHE) with circuit privacy~\cite{OPCPC14}.
Our security proofs hold in the Quantum Random Oracle Model (QROM)~\cite{BDFLSZ11}. For simplicity, in our exposition we assume that the relevant security parameters are polynomial in the input \BQP instance size $n$, so that efficient algorithms run in time $\poly(n)$ and errors are (ideally) negligible in $n$.

\paragraph{Warmup: A non-interactive test of quantumness.}\label{NI-test}

To explain our approach, we first briefly describe how to make the ``cryptographic test of quantumness'' (CTQ) of~\cite{BCMVV18} into a non-interactive protocol (with setup.) This is a significantly simplified version of the Mahadev protocol: there is no ground state, and the initial state \eq{mah-sketch} is simply in uniform superposition over $X$. The soundness error is $1/2$, meaning that a \emph{classical} prover can convince the verifier to accept with probability at most $1/2$~\cite{BCMVV18}. A quantum prover can easily answer both challenges, so the completeness is~$1$.

To reduce the interaction in this protocol, we perform two transformations. First, we repeat the protocol independently in parallel $k$ times, with the verifier accepting if and only if all $k$ copies accept. We then remove Round 3 via the Fiat-Shamir transform~\cite{FS86}: the prover computes the challenges $\mathbf c = (c_1, c_2, \dots, c_k) := \mathcal H(y_1, \dots, y_k)$ themselves via a public hash function $\mathcal H$. This allows the prover to go directly to Round 4, i.e., measuring the $X$ registers. The verifier then performs the $k$-fold accept/reject verdict calculations, using coins $\mathbf c$ computed in the same manner. The result is a two-message protocol. Moreover, since the keys are drawn from a fixed distribution, we can give the $pk_j$ to the prover and the $sk_j$ to the verifier in an offline setup phase, so that the protocol only requires one message in the online phase. We refer to this protocol as NI-CTQ.

Since the soundness experiment of NI-CTQ only involves classical provers, and the verifier is also classical, soundness can be deduced from existing classical results. Specifically, standard parallel repetition\footnote{A subtlety is that this is a private-coin protocol. However, the $c=0$ branch  is publicly simulable so~\cite{HPWP10} applies. Alternatively, one can apply the techniques of~\cite{Hai09}.} theorems~\cite{Hai09,HPWP10,BHT19} combined with soundness of Fiat-Shamir~\cite{BR95,FS86,PS00} yield the fact that NI-CTQ has negligible soundness and completeness errors, in the Random Oracle Model (ROM).

\paragraph{Transforming the Mahadev protocol.}

Similar to NI-CTQ above, we will apply various transformations to the Mahadev verification protocol itself:
\begin{enumerate}
\item making the first message instance-independent (i.e., moving it to an offline setup phase);
\item applying parallel repetition, via a new parallel repetition theorem;
\item adding zero-knowledge, by means of classical NIZKs and classical FHE; and
\item applying Fiat-Shamir (in the QROM~\cite{BDFLSZ11}).
\end{enumerate}
Unlike with NI-CTQ, however, establishing that these transformations satisfy desirable properties is much more challenging. For instance, since cheating provers can now be quantum, classical parallel repetition theorems do not apply.

\paragraph{Instance-independent setup.}
The first transformation is relatively simple to describe, at a high level.
In the Mahadev protocol, the initial message depends on a sequence of basis choices ($X$ or $Z$) for measuring the ground state of a $\mathsf{ZX}$ Hamiltonian. These choices need to be consistent with a particular two-local term drawn from some distribution $\mathcal D$.
Clearly, a \emph{random} choice is correct with probability $1/4$. Now, if we consider multiple copies of the ground state, and each copy is assigned both a random choice of bases and a random term from $\mathcal D$, then about $1/4$ of the copies get a consistent assignment. We can then make the initial message instance-independent by increasing the number of copies of the ground state in the Mahadev protocol by a constant factor. We establish this fact\footnote{More precisely, we apply this transformation at the level of the Morimae-Fitzsimons protocol~\cite{MF16}, an important building block of the Mahadev protocol.} in \expref{Lemma}{lem:modified-mf} below. We refer to the result as ``the three-round Mahadev protocol,'' and denote it by $\mathfrak M$.

\paragraph{Parallel repetition.}
The $k$-fold sequential repetition of a protocol is a simple way of decreasing the original soundness error $\delta$ to $\delta^k$, at the cost of multiplying the number of interaction rounds by $k$. Parallel repetition is much more desirable because it does not increase the number of rounds. However, even in the case of purely classical protocols, proving that parallel repetition reduces soundness error is often quite difficult, and may require adapting the protocol itself~\cite{BIN97,Hai09}.

Does parallel repetition work for quantum-prover interactive arguments? The Mahadev protocol is a natural case to consider since it already exhibits the full decisional power of QPIAs, namely $\BQP$. However, several complications arise when attempting to establish parallel repetition using classical techniques. First, the Mahadev protocol is clearly private-coin, precisely the category that is challenging even classically~\cite{BIN97,Hai09}. Second, classical proofs of parallel repetition typically involve constructing a prover (for the single-copy protocol) that uses many rounds of nested rejection sampling. The quantum analogue of such a procedure is quantum rewinding,
which can only be applied in special circumstances \cite{Wat09,ARU14} and seems difficult to apply to parallel repetition.

In this work, we establish a new parallel repetition theorem with alternative techniques, suited specifically for the Mahadev protocol. We show that, for NO instances, the accepting paths of the verifier for the two different challenges ($c=0$ and $c=1$) correspond to two nearly (computationally) orthogonal projectors. We also establish that this persists in $k$-fold parallel repetition, meaning that each pair of distinct challenge strings $\mathbf c, \mathbf c' \in \{0,1\}^k$ corresponds to nearly orthogonal projectors. From there, a straightforward argument shows that the prover cannot succeed for more than a non-negligible fraction of challenge strings. Our result shows that $k$-fold parallel repetition yields the same optimal soundness error $\delta^k$ as sequential repetition.

Taken together with the first transformation, the result is a three-round \QPIA (with offline setup) for verifying $\BQP$, with negligible completeness and soundness errors. We denote this protocol by $\mathfrak M^k$.

\begin{theorem}\label{thm:parallel-intro}
Under the \LWE assumption, the $k$-fold parallel repetition $\mathfrak M^k$ of the three-round Mahadev protocol $\mathfrak M$ is a three-round protocol (with offline setup) for verifying $\BQP$ with completeness $1 - \negl(n)$ and soundness error $2^{-k} + \negl(n)$.
\end{theorem}

\paragraph{Zero-knowledge.}

Zero-knowledge is a very useful cryptographic property of proof systems. Roughly, a protocol is zero-knowledge if the verifier ``learns nothing'' from the interaction with the honest prover, besides the fact that the relevant instance is indeed a ``yes'' instance. This notion  is formalized by requiring an efficient simulator whose output distribution is indistinguishable from the distribution of the outcomes of the protocol.

In our next result, we show how to modify the protocol $\mathfrak{M}^k$ of \expref{Theorem}{thm:parallel-intro} to achieve zero-knowledge against arbitrary classical verifiers. Our approach is similar to that of~\cite{CVZ19}, but uses a purely classical verifier. Instead of the prover providing the outcomes of the measurements to be checked by the verifier (as in $\mathfrak{M}^k$), a classical non-interactive zero-knowledge proof (NIZK) is provided. However, the \NP statement ``the measurements will pass verification'' depends on the inversion trapdoor of the verifier, which must remain secret from the prover. To overcome this obstacle, we use classical fully homomorphic encryption (FHE). In the setup phase, an encryption of the verifier's secret keys is provided to the prover, enabling the prover to later compute the NIZK homomorphically. To establish the zero-knowledge property, we require the FHE scheme to have circuit privacy, which means that the verifier cannot learn the evaluated circuit \emph{from the ciphertext} provided by the prover. To prove the zero-knowledge property, we also need the extra assumption that the setup phase is performed by a trusted third party, since we cannot rely on the verifier to perform it honestly anymore.

In classical zero-knowledge arguments, it is common to consider efficient provers who are provided an $\NP$-witness of the statement to prove. In the quantum setting, if we assume that the quantum polynomial-time prover has access to a quantum proof of a $\QMA$ statement,\footnote{$\QMA$ is a decision problem class which is a quantum analogue of $\NP$. In $\QMA$, an untrusted quantum proof is provided to a quantum poly-time verifier.} we achieve the following.

\begin{theorem}[informal]\label{thm:zk-intro}
Under the \LWE assumption, if circuit-private FHE exists, then there exists a three-round zero-knowledge argument for $\QMA$ (with trusted setup) with negligible completeness and soundness error. \end{theorem}

\paragraph{Fiat-Shamir transformation.}

Note that in the protocols discussed above (both $\mathfrak{M}^k$ and its ZK-variant), the second message of the verifier to the prover is a uniformly random $\mathbf c \in \{0,1\}^k$. In the final transformation, we eliminate this ``challenge'' round. This is done via the well-known Fiat-Shamir transform~\cite{FS86}: we ask the prover to generate the challenge bits $\mathbf c \in \{0,1\}^k$ themselves by evaluating a public hash function $\H$ on the transcript of the protocol thus far. In our case, recalling \eq{mah-sketch}, this means that the prover selects $\mathbf c := \H(H, pk, y)$. Note that $pk$ and $y$ are now both $k$-tuples, since we are transforming $k$-fold parallel-repeated protocols. Of course, the verifier also needs to adapt their actions at the verdict stage, using $\mathbf c = \H(H, pk, y)$ when deciding whether to accept or reject. The resulting protocols now only have a setup phase and a single message from the prover to the verifier.

A standard approach with Fiat-Shamir (FS) is to establish security in the Random Oracle Model, in the sense that FS preserves soundness up to a loss that is negligible provided $\H$ has a superpolynomially-large range \cite{BR95,PS00}. It is straightforward to see that this last condition is required; it is also the reason that we applied parallel repetition prior to FS. A well-known complication in the quantum setting is that quantum computers can evaluate any public classical function $\H$ in superposition via the unitary operator $U_{\H} \colon \ket{x}\ket{y} \mapsto \ket{x}\ket{y \oplus \H(x)}.$ This means that we must work in the Quantum Random Oracle Model (QROM) \cite{BDFLSZ11}, which grants all parties oracle access to $U_{\H}$. Proving the security of transformations like FS in the QROM is the subject of recent research, and newly developed techniques have largely shown that FS in the QROM preserves soundness for so-called $\Sigma$-protocols~\cite{DFMS19,LZ19}. Extending those results to our protocols is relatively straightforward. Applying FS to the three-round verification protocol from \expref{Theorem}{thm:parallel-intro} then yields the following.

\begin{theorem}\label{thm:noninteractive-intro}
Let $k=\omega(\log n)$, and let $\FS(\mathfrak{M}^k)$ denote the protocol resulting from applying Fiat-Shamir to the $k$-fold parallel repetition of the three-round Mahadev protocol. Under the \LWE assumption, in the QROM, $\FS(\mathfrak{M}^k)$ is a non-interactive protocol (with offline setup) for verifying $\BQP$ with negligible completeness and soundness errors. \end{theorem}
If we instead apply the Fiat-Shamir transform to the zero-knowledge protocol from \expref{Theorem}{thm:zk-intro}, we achieve the following.
\begin{theorem}[informal]\label{thm:nizk-intro}
Under the \LWE assumption, in the QROM, there exists a classical non-interactive zero-knowledge argument (with trusted offline setup) for $\QMA$, with negligible completeness and soundness errors.
\end{theorem}

\paragraph{Related results.}

Broadbent, Ji, Song, and Watrous \cite{BJSW16} presented the first quantum zero-knowledge proofs for $\QMA$ with efficient provers. Vidick and Zhang~\cite{VZ19} combined this protocol with the Mahadev protocol~\cite{Mah18} to make the communication classical. Broadbent and Grilo~\cite{BG19} showed a ``quantum $\Sigma$'' zero-knowledge proof for $\QMA$ (with a quantum verifier). In the non-interactive setting, Coladangelo, Vidick, and Zhang~\cite{CVZ19} constructed a non-interactive zero-knowledge argument with quantum setup and Broadbent and Grilo~\cite{BG19} showed a quantum statistical zero-knowledge proof in the secret parameter model.

We remark that Radian and Sattath~\cite{RS19} recently established what they call ``a parallel repetition theorem for NTCFs,'' which are the functions $f_{pk}$ in the Mahadev protocol. However, the context of~\cite{RS19} is very different from that of our \expref{Theorem}{thm:parallel-intro}. They work with 1-of-2 puzzles, not \BQP verification; in particular, their soundness experiment is quite different. Moreover, their parallel repetition theorem follows from a purely classical result.

After an initial version of our work was made public, showing how the Mahadev protocol can be reduced to four rounds using parallel repetition and the Fiat-Shamir transform, Chia, Chung, and Yamakawa posted a preprint~\cite{CCY19} describing the same result, with an alternative proof of parallel repetition. They also showed how to make the verifier run in time polylogarithmic in the instance size using indistinguishability obfuscation. Our work was performed independently, and we subsequently improved our result to make the protocol non-interactive with setup and zero-knowledge.

\paragraph{Open problems.}

This work raises several natural open questions. First, is it possible to prove the soundness of our protocol when the oracle $\H$ is instantiated with a concrete (e.g., correlation-intractable~\cite{PS19}) hash function? Our current analysis only applies in an idealized model.

It is also natural to study parallel repetition for general QPIAs. Natural examples include the protocols of~\cite{GV19,VZ19,BCMVV18}. It is known that parallel repetition does not reduce the soundness error even for all classical private-coin protocols~\cite{BIN97}. However, if the classical protocol is slightly modified to to include ``random terminations'' with low probability, then parallel repetition is in fact possible~\cite{Hai09,HPWP10,BHT19}. It is an open question whether similar mild relaxations can enable a parallel repetition theorem for any QPIA.

Finally, we remark that a classical NIZK protocol (in the Random Oracle Model with setup) could also be achieved using the techniques of locally simulatable codes/proofs~\cite{GriloSY19,BG19}. We leave as an open problem understanding whether such a protocol could give us useful properties that are not achieved with our current approach.

\paragraph{Organization.}

The remainder of the paper is organized as follows.
In \sec{preliminaries}, we introduce QPIAs, the local Hamiltonian problem as it relates to $\BQP$ verification, the Mahadev protocol, and non-interactive zero knowledge.
In \sec{indep}, we explain how to make the initial step of the protocol instance-independent.
In \sec{parallel}, we show that parallel repetition reduces the soundness error of the Mahadev protocol at the optimal rate.
In \sec{zk}, we describe how to make the protocol zero-knowledge.
Finally, in \sec{FS}, we show that under the Fiat-Shamir transformation, a generalization of $\Sigma$-protocols (which includes the protocols of interest to us) remains secure in the QROM, and we use this to establish non-interactive protocols (with offline setup) for verifying $\BQP$ and for zero-knowledge verification of $\QMA$.

\section{Preliminaries}\label{sec:preliminaries}

\subsection{Notation and conventions}

Most algorithms we consider are efficient, meaning that they run in time polynomial in both the input size (typically $n$) and the security parameter (typically $\lambda$). We assume that $n$ and $\lambda$ are polynomially-related. The two main classes of algorithms of interest are \PPT (probabilistic poly-time) and \QPT (quantum poly-time). We say that  $f = \negl(n)$ if $f = o(n^{-c})$ for every constant $c$. We denote by $U_f$ the efficient map that coherently implements a classical function $f\colon \bit^n \to \bit^m$, i.e., $U_f\ket{x}\ket{y} = \ket{x}\ket{y \oplus f(x)}$, when there exists an efficient deterministic circuit that computes $f$. We make use of commitment schemes and fully-homomorphic encryption; for completeness we define them in \app{standard-primitives}.

\subsection{Quantum-prover interactive arguments}

A quantum-prover interactive argument (\QPIA) is an interactive protocol between two polynomially-bounded parties, a quantum prover and a classical verifier, interacting over a classical channel. A \QPIA is described by a pair of algorithms: the \PPT algorithm of the honest verifier \verifier, and the \QPT algorithm of the honest prover \prover.

\begin{definition}\label{dfn:qpia}
Fix a language $L \subseteq \{0,1\}^*$ and a \QPIA $(\prover, \verifier)$.
  We say that $(\prover, \verifier)$ is a \QPIA for $L$ with \emph{completeness} $c$  and \emph{soundness error} $s$ if
\begin{itemize}[nosep]
\item for all $x \in L$, $\Pr[\text{\prover and \verifier accept $x$}] \geq c$.
\item for all $x \notin L$ and for all \QPT algorithms $\prover'$, $\Pr[\text{$\prover'$ and \verifier accept $x$}] \leq s$.
\end{itemize}
  The \emph{completeness error} is $1-c$ and the \emph{soundness} is $1-s$.
\end{definition}

As discussed above, the soundness of a \QPIA can be amplified via sequential repetition, and in some cases also with parallel repetition.

\subsection{The local Hamiltonian problem and verification for $\BQP$}\label{sec:local-H}

Any promise problem $L=(L_{\yes},L_{\no})\in\QMA$ can be reduced to the local Hamiltonian problem such that for $x\in L_{\yes}$, the Hamiltonian $H_x$ has a low-energy ground state $\ket{\psi_x}$, and for $x\in L_{\no}$, all quantum states have large energy \cite{KSV02}.
While the quantum witness $\ket{\psi_x}$ may be hard to prepare for general $L \in \QMA$, it can be prepared efficiently if $L\in\BQP$. Furthermore, the
problem remains QMA-complete even with a Hamiltonian that can be measured
by performing standard ($Z$) and Hadamard ($X$) basis measurements \cite{BL08,CM16}.

\begin{problem}[The 2-local ZX-Hamiltonian problem \cite{BL08,CM16,MF16}]\label{prob:ZX}
  The 2-local ZX-Hamiltonian promise problem $\zx=(\zx_{\yes},\zx_{\no})$, with parameters $a,b\in \mathbb{R}$, $b>a$ and gap $(b-a)>\poly(n)^{-1}$, is defined as follows. An instance is a local Hamiltonian
\begin{align*}
  H = \sum_{i<j} J_{ij} (X_iX_j + Z_iZ_j)
\end{align*}
where each $J_{ij}$ is a real number
  such that $2\sum_{i<j} |J_{ij}| = 1$
and each $X_i$ (resp.\ $Z_i$) is a Pauli $X$ (resp.\ Pauli $Z$) gate acting on the $i$th qubit.
For $H\in \zx_{\yes}$, the smallest eigenvalue of $H$ is at most $a$, while if $H\in \zx_{\no}$, the smallest eigenvalue of $H$ is at least $b$.
\end{problem}

Note that given the normalization factors, we can see that each term ($X_iX_j$ or $Z_iZ_j$) is associated with the probability $J_{ij}$.
We denote the 2-local ZX-Hamiltonian problem with parameters $a,b\in\mathbb{R}$ by $\zx_{a,b}$.

When working with Hamiltonian terms $S$, we overload the notation for convenience. First, we write $S_j$ to denote the Pauli operator assigned by $S$ to qubit $j$, so that $S = \bigotimes_j S_j$. Second, we write $i \in S$ to indicate that $i$ is a qubit index for which $S$ does not act as the identity, i.e., $S_i \neq \1$.

Morimae and Fitzsimons present a protocol (the ``MF protocol'')
with a quantum prover and a limited verifier who only needs the ability to perform single-qubit $X$ and $Z$ basis measurements \cite{MF16}. The prover $\prover$ prepares the ground state of the Hamiltonian and sends it to $\verifier$, who then samples a term $S$ with probability $p_S$ and performs the corresponding measurement. Notice that to estimate the energy of term $S$, only $Z$ or $X$ basis measurements are necessary. In the original protocol of~\cite{MF16}, the qubits are sent individually; in the variant below, we simply have the prover send the entire state all at once.

\begin{protocol}[A variant of the MF protocol]\label{prot:MF}~
\begin{proto}
\item[Setup.] $\prover$ and $\verifier$ receive an instance of \prob{ZX}, namely a Hamiltonian $H_x=\sum_S p_S\frac{\1+m_S S}{2}$, where each $S$ is a tensor product of $X,Z,\1$ and $m_S \in \{\pm 1\}$.
\item[Round~1.] $\prover$ prepares a quantum state $\rho$ and sends it to $\verifier$.
\item[Verdict.] $\verifier$ samples a term $S$ with probability $p_S$ and performs the measurement $\{M_1=\frac{\1+S}{2},M_{-1}=\frac{\1-S}{2}\}$ on $\rho$, getting an outcome $e$.
  $\verifier$ accepts if $e=-m_S$.
\end{proto}
\end{protocol}

Since $M_{m_S}=\frac{\1+m_S S}{2}=\1-M_{-m_S}$ and $H_x=\sum_S p_S M_{m_S}$, the success probability of the protocol with input state $\rho$ is
\begin{align*}
  \sum_S p_S \tr(M_{-m_S} \rho)
  &= 1 - \sum_s p_S \tr(M_{m_s} \rho)
  = 1-\tr(H_x\rho).
\end{align*}
Since $b-a>\poly(|x|)^{-1}$, the error in the MF protocol can be made negligible by parallel repetition: $\verifier$ receives $T$ copies of the ground state of $H$ and performs an independent test on each copy. By accepting if at least $(2-a-b)T/4$ copies accept, both the completeness and soundness errors are suppressed to negligible with polynomial $T(|x|)$ (cf.~\cite[Theorem~8.4]{Mah18}). For a detailed proof of $\QMA$ gap amplification, see \cite[Section~3]{VW16}.

In the following discussion, the term $S$ is encoded by an $n$-bit string $h(S)$:
for each qubit $i\in S$, set $h_i=0$ for a $Z$ basis measurement and $h_i=1$ for an $X$ basis measurement.
For other qubits, the choice is irrelevant but we set $h_i=0$ for concreteness. We let
$\alpha_{h,\rho} := \tr(M_{-m_S}\rho)$
denote the success probability with $\rho$ when $h=h(S)$ is sampled in \prot{MF}.

\subsection{The Mahadev protocol for $\BQP$ verification}\label{sec:measurement-protocol}

\paragraph{Required primitives.}

The protocol relies crucially on two special classes of functions: Noisy Trapdoor Claw-free Functions (NTCFs) $\F$ and Noisy Trapdoor Injective Functions (NTIFs) $\G$. Both classes of functions are constructed based on the presumed hardness of the Learning with Errors (LWE) problem \cite{BCMVV18,Mah18}. We now sketch the properties of these function families. For complete details, and for the LWE construction, see~\cite{BCMVV18}. Let $\lambda$ be a security parameter and let $q\geq 2$ be prime. Choose parameters $\ell=\poly(\lambda)$, $n=\Omega(\ell\log q)$, and $m=\Omega(n\log q)$.

The NTCF family $\F=\{f_{pk}\}_{pk\in\K_\F}$ is a family of keyed functions
\begin{align*}
f_{pk}\colon \{0,1\}\times\X\to D_\Y
\end{align*}
which, on input a public key $pk\in\K_\F:=\ZZ_q^{m\times n}\times\ZZ_q^m$, a bit $b$, and $x\in\X:=\ZZ_q^n$, outputs a distribution $f_{pk}(b,x)$ over $\Y:=\ZZ_q^m$.
Each function $f_{pk}\in\F$ satisfies the \emph{injective pair} property: there exists a perfect matching $\mathcal R_{pk} \subset\X\times\X$ such that $f_{pk}(0, x_0) = f_{pk}(1, x_1)$ if and only if $(x_0, x_1) \in \mathcal R_{pk}$.

The NTCF family is equipped with the following polynomial-time algorithms:
\begin{enumerate}
  \item $\Gen_\F$, on input $1^\lambda$, outputs a secret-public key pair $(pk,sk)$.
  \item $\Chk_\F$ is a deterministic algorithm for checking if $(b,x)$ and $y$ form a preimage-image pair of $f_{pk}$. On input $b,x,y$, $\Chk_\F$ outputs 1 iff $y\in\supp(f_{pk}(b,x))$.
  \item $\Inv_\F$ is a deterministic algorithm for inverting the function $f_{pk}$. On input secret key $sk$, bit $b$, and image $y$, $\Inv_\F$ returns the preimage $x_{b,y}$ such that $y\in\supp(f_{pk}(b,x_{b,y}))$, or outputs $\mathsf{reject}$ if no such preimage exists.
  \item $\Samp_\F$ is an efficient quantum process which,
    on input $pk$ and $b \in \{0,1\}$, returns a quantum state negligibly close to
\begin{align}\label{eq:samp}
  \frac{1}{|\X|^{1/2}}\sum_{x\in\X}\ket{b}\ket{x}\ket{\psi_{f_{pk}(b,x)}},
\end{align}
where $\ket{\psi_{p}}:=\sum_{y\in\Y} \sqrt{p(y)}\ket{y}$ for distribution $p$.
By the injective pair property,
we have $\langle \psi_{f_{pk}(b,x)}|\psi_{f_{pk}(b',x')}\rangle =1$ if $(b,x)=(b',x')$, or there exists $(x_0,x_1)\in \R_{pk}$ such that $(b,x,b',x')=(0,x_0,1,x_1)$ or $(1,x_1,0,x_0)$.
This implies that the states in $\{\ket{\psi_{f_{pk}(b,x)}}\}$ can be perfectly distinguished by performing a standard basis measurement.
Thus, intuitively, we may consider an \emph{ideal} version of these functions, i.e., the distribution $p$ is concentrated at a single point.
\end{enumerate}

Similarly, the NTIF family $\G=\{g_{pk}\}_{pk\in\K_\G}$ is a family of keyed functions
\begin{align*}
g_{pk}\colon \{0,1\}\times\X\to D_{\Y}
\end{align*}
which, on input a public key $pk\in\K_\G$, a bit $b$, and $x\in\X$, outputs a distribution $g_{pk}(b,x)$ over $\Y$. Instead of the injective pair property of NTCFs, NTIFs satisfy an \emph{injectivity} property: for all $(x, b) \neq (x', b')$, $\supp g_{pk}(b, x) \cap \supp g_{pk}(b', x') = \emptyset.$ An NTIF family is also equipped with a tuple of four polynomial-time algorithms $(\Gen_\G,\Chk_\G,\Inv_\G,\Samp_\G)$, defined exactly as in the NTCF case (but with $g$ in place of $f$, and $\G$ instead of $\F$.)

We remark that the states \eq{samp} prepared by $\Samp_\F$ and $\Samp_\G$ should be contrasted with the ``idealized'' state described in \eq{mah-sketch} in our sketch of the protocol.

\paragraph{The protocol.}

The Mahadev protocol \cite{Mah18} for $\BQP$ verification allows $\verifier$ to request an $X$ or $Z$ basis measurement outcome without revealing the basis to $\prover$.
The aim of the protocol is to verify that the prover's response, when appropriately decoded, is close to the measurement outcomes of some $n$-qubit quantum state $\rho$. Crucially, this guarantee holds simultaneously for all basis choices $h\in\{0,1\}^n$, where $0$ denotes a $Z$ basis measurement and $1$ denotes an $X$ basis measurement. With this guarantee, the verifier can then apply the verification procedure of the MF protocol to the decoded responses of the prover, knowing that this will correctly decide whether the instance should be accepted or rejected.

In the following protocol, for each qubit, if $\verifier$ requests a $Z$ basis measurement, then an NTIF key is sent;
if $\verifier$ requests an $X$ basis measurement, then an NTCF key is sent. Since $\Chk_\F$ and $\Chk_\G$ are identical, we denote them by $\Chk$.
Similarly, $\Samp_\F$ and $\Samp_\G$ are identical, so we denote them by $\Samp$.
We let $\Gen(1^\lambda,h)$ for $h\in\bit^*$ denote the following key generation algorithm:
for every bit $i$ of $h$, run $(pk_i,sk_i)\gets\Gen_\G(1^\lambda)$ if $h_i=0$ and $(pk_i,sk_i)\gets\Gen_\F(1^\lambda)$ if $h_i=1$.
Set $pk=(pk_i)_i$ and $sk=(sk_i)_i$ and output the key pairs $(pk,sk)$.

We now the describe the protocol when $\verifier$ and $\prover$  are honest.

\begin{protocol}[Mahadev protocol]\label{prot:Mahadev}~
\begin{proto}
\item[Setup.]
  Choose a security parameter $\lambda \geq n$. Both $\prover$ and $\verifier$ receive an instance of \prob{ZX}, namely $H=\sum_S p_S\frac{\1+m_S S}{2}$.
\item[Round~$\verifier_1$.]
  $\verifier$ samples $r$ terms $S=(S_1,\ldots,S_r)$ and computes $h=h(S)$.
  $\verifier$ generates the key pair $(pk,sk)\gets\Gen(1^\lambda,h)$ and sends $pk$ to $\prover$.
\item[Round~$\prover_1$.]
$\prover$ prepares $r$ copies of the $n$-qubit ground state $\ket{\phi}^{\otimes r}=\sum_{b\in\bit^n}\phi_b\ket{b}_W$ of $H$ in register $W$.
For $j\in[r],\ell\in[n]$ and each qubit $W_{j\ell}$ in $W$, $\prover$ performs $\Samp$ on input the key $pk_i$ coherently and yields a quantum state negligibly close to
\begin{align*}
\frac{1}{|\X|^{n/2}}\sum_{x\in\X^n} \sum_{b \in \{0,1\}^{nr}} \phi_{b}\ket{b}_W\ket{x}_X\ket{\psi_{f_{pk}(b,x)}}_Y\,.
\end{align*}
Next, $\prover$ measures $Y$ and sends the outcome $y$ to $\verifier$.

\item[Round $\verifier_2$.] $\verifier$ responds with a uniformly random ``challenge'' bit $c \in \bit$.  We call $c=0$ ``test round'' and $c=1$ ``Hadamard round,'' and set labels $\test = 0$ and $\hada = 1$.
\item[Round $\prover_2$.] If $c=\test$, $\prover$ measures $WX$ in the computational basis.
If $c=\hada$, $\prover$ measures $WX$ in the Hadamard basis.
  In either case, $\prover$ sends the measurement outcome $(w, t)$ to $\verifier$;

\item[Verdict.]
      In a test round, $\verifier$ accepts if $\bigwedge_{j\in[r],\ell\in[n]}\Chk(pk_{j\ell},w_{j\ell},t_{j\ell},y_{j\ell})=1$.

      In a Hadamard round, $\verifier$ performs the following: for each copy $j$ and qubit $\ell\in[n]$,
      \begin{enumerate}
        \item \label{h1-check} if $h_{j\ell}=0$, run $(e_{j\ell},x_{b_{j\ell},y_{j\ell}})\gets\Inv_\G(sk_{j\ell},y_{j\ell})$.
          If $h_{j\ell}=1$, run $x_{0,y_{j\ell}}\gets\Inv_\F(sk_{j\ell},0,y_{j\ell})$ and $x_{1,y_{j\ell}}\gets\Inv_\G(sk_{j\ell},1,y_{j\ell})$ to get both preimages.
          Then compute $e_{j\ell}=t_{j\ell}\cdot(x_{0,y_{j\ell}}\oplus x_{1,y_{j\ell}})\oplus w_{j\ell}$.

          If any of the above the algorithms returns $\mathsf{reject}$ for any $j,\ell$, or $t_{j\ell}$ is trivial (e.g., $0^n$, see~\cite{BCMVV18}), reject.
          Otherwise store $e=(e_{j\ell})_{j\in[r],\ell\in[n]}$ as the witness to the next check.
        \item \label{h2-check} $\verifier$ sets $v_{j}=1$ if the witness $(e_{j\ell})_{\ell\in S_j}$ satisfies $M_{-m_{S_{j}}}$ (defined in \prot{MF}).

      \end{enumerate}
      Finally, $\verifier$ accepts if $\sum_{j\in[r]} v_j \geq (2-a-b)r/4$.
\end{proto}
\end{protocol}

\begin{theorem}[Theorems 1.1 and 8.6 in~\cite{Mah18}]\label{thm:mah-main}
Under the \LWE assumption, \prot{Mahadev} is a four-message quantum-prover interactive argument for the class \BQP with completeness error $\negl(n)$ and soundness error $3/4 + \negl(n)$.
\end{theorem}

\subsection{NIZK for NP}

We will use NIZK protocols $\NIZK=(\Setup,\aP,\aV,\aS)$ for \NP. The protocol is described in terms of the following algorithms.
\begin{enumerate}
  \item (Setup) $\Setup(1^\lambda)$ outputs a common reference string $\crs$ on input $1^\lambda$.
  \item (Prove) $\aP(\crs,x,u)$ outputs a string $e$ on input $\crs$, instance $x$, and witness $u$.
  \item (Verify) $\aV(\crs,x,e)$ outputs a bit $b\in\bit$ on input $\crs$, instance $x$, and proof $e$.
  \item (Simulate) $\aS(x)$ outputs a transcript $(\crs,x,e)$.
\end{enumerate}
We use the Peikert-Shiehian construction based on \LWE~\cite{PS19}. For any \NP language $\LL$, the protocol satisfies the following:
\begin{enumerate}
  \item Completeness: for $(x,u)\in\R_\LL$,
    \begin{align*}
      \Pr\left[\aV(\crs,x,e)=1\middle|\begin{aligned} \crs &\gets\Setup(1^\lambda) \\ e &\gets\aP(\crs,x,w) \end{aligned}\right]\geq 1-\negl(n).
    \end{align*}

  \item Adaptive soundness: for any quantum adversary $\A$, 
    \begin{align*}
      \Pr\left[x\notin\LL\wedge\aV(\crs,x,e)=1\middle|\begin{aligned} \crs &\gets\Setup(1^\lambda) \\ (x,e) &\gets\A(\crs) \end{aligned}\right]\leq\negl(n).
    \end{align*}
  \item Zero-knowledge: for $(x,u)\in\R_\LL$, the distributions $\{(\crs,\aP(\crs,x,u))\}$ and $\{\aS(x)\}$ are computationally indistinguishable.
\end{enumerate}

\begin{remark}
  In a \textbf{classical} NIZK for $\QMA$, there are two differences from the above: (i.) the witness provided to the prover is a quantum state $\ket{\psi}$ instead of the classical value $u$, (ii.) $\aP$ is a quantum polynomial-time algorithm, with classical output, and (iii.) our NIZK protocol for $\QMA$ works in the random oracle model with setup, instead of the (desired) common reference string model.
\end{remark}

\section{Instance-independent key generation}\label{sec:indep}

In this section, we modify the MF protocol such that the sampling of the Hamiltonian term is independent of the performed measurements. This change allows the keys in the Mahadev protocol to be generated before the parties receive the input Hamiltonian, in an offline setup phase.

In our variant, for some $r = \poly(n)$, the verifier $\verifier$ samples $r$ $n$-bit strings $h_1,\ldots,h_r$ uniformly and $r$ independent 2-local terms $S_1,\ldots,S_r$ according to distribution $\pi$ (in which $S$ is samplied with the probability $p_S$ from \expref{Protocol}{prot:MF}).
We say the bases $h_i$ and the terms $S_i$ are \emph{consistent} if, when the observable for the $j$th qubit in $S_i$ is $Z$ (resp., $X$) then the $j$th bit of $h_i$ is $0$ (resp., $1$).
Since $h_i$ is uniformly sampled and $S_i$ is $2$-local, we have
\begin{align*}
  \Pr_{S_i\gets\pi, h_i\gets\bit^n}[\text{$S_i$ and $h_i$ are consistent}] \geq\frac{1}{4}.
\end{align*}
In an $r$-copy protocol, we let $A:=\{i\in[r]: \text{$h_i$ and $S_i$ are consistent}\}$ and denote $t=|A|$.
For each $i \in A$, $\verifier_i$ decides as in the MF protocol: if $i \notin A$, then $\verifier_i$ accepts.
Thus we consider the following protocol.
\begin{protocol}[A modified parallel-repeated MF protocol for $\zx_{a,b}$]\label{prot:modified-MF}~
  \begin{proto}
    \item[Setup.] $\verifier$ samples the bases $h_1,\ldots,h_r\gets\bit^n$ uniformly.
    \item[Round~1.] $\prover$ sends the witness state $\rho$ ($r$ copies of the ground state).
    \item[Round~2.]
      $\verifier$ measures the quantum state $\rho$ in the bases $h_1,\ldots,h_r$.
      For each copy $i\in[r]$, $\verifier$ samples terms $S_1,\ldots,S_r\gets\pi$.
      $\verifier$ records the subset $A\subseteq [r]$ of consistent copies.
      For each copy $i\in A$, $\verifier$ sets $v_i=1$ if the outcome satisfies $M_{-m_S}$ and 0 otherwise.
      $\verifier$ accepts if $\sum_{i\in A} v_i\geq (2-a-b)|A|/4$.
  \end{proto}
\end{protocol}
For sufficiently large $r$, with high probability, there are around $r/4$ consistent copies. Thus to achieve the same completeness and soundness, it suffices to increase the number of copies by a constant factor. We thus have the following fact, proved formally in \app{modifiedmf}.

\begin{restatable}{lemma}{modifiedmf}\label{lem:modified-mf}
  The completeness error and soundness error of \prot{modified-MF} are negligible, provided $r=\smash{\omega\bigl(\frac{\log n}{(b-a)^2}\bigr)}$ copies are used.
\end{restatable}
\begin{remark}\label{rmk:alpha}
We stress that the terms $S_i$ are sampled {\em independently} of the interaction in the protocol. We let $\term(H,s)$ denote the deterministic algorithm that outputs a term from $H$ according to distribution $\pi$ when provided the randomness $s\in\bit^p$ for sufficiently large polynomial $p$. For bases $h\in\bit^{nr}$ and $s\in\bit^p$, we let $\alpha_{h,s,\rho}$ denote the success probability when $\prover$ presents quantum state $\rho$.
\end{remark}

We embed \prot{modified-MF} into our classical verification protocol by modifying the Mahadev protocol (\prot{Mahadev}) as follows.
First, note that the measurement in random bases $h$ can be achieved by sending random keys, so the key generation can be done before the parties receive the instance.
In the Verdict stage, if the protocol enters a test round (i.e., $c=0$), then $\verifier$ checks as in \prot{Mahadev}.
If the protocol enters a Hadamard round (i.e., $c=1$), $\verifier$ uses the secret key to compute the measurement outcome, as in check \ref{h1-check} of the Verdict stage of \prot{Mahadev}.
Once the outcomes are successfully computed, $\verifier$ samples terms $S_1,\ldots,S_r\gets\pi$ and checks the consistent copies.
Since the outcome must be computationally indistinguishable from measuring a quantum state in bases $h$,
\lem{modified-mf} applies, and \thm{mah-main} holds for our variant protocol. We refer to this variant as ``the three-round Mahadev protocol'' and denote it by $\mathfrak{M}$.

\section{A parallel repetition theorem for the Mahadev protocol}\label{sec:parallel}

In a $k$-fold parallel repetition of $\mathfrak{M}$, the honest prover runs the honest single-fold prover independently for each copy of the protocol. Meanwhile, the honest verifier runs the single-fold verifier independently for each copy, accepting if and only if all $k$ verifiers accept. The completeness error clearly remains negligible. We now analyze the soundness error and establish a parallel repetition theorem.

In preparation, we fix the following notation related to the Verdict stage of $\mathfrak{M}$. We will refer frequently to the notation established in our description of \expref{Protocol}{prot:Mahadev} above, which applies to $\mathfrak{M}$ as well. First, the check $\bigwedge_{j\in[r],\ell\in[n]}\Chk(pk_{j\ell},w_{j\ell},t_{j\ell},y_{j\ell})=1$ in a test round is represented by a projection $\Pi_{sk,\test}$ acting on registers $WXY$.
Specifically, this is the projector whose image is spanned by all inputs $(w,t,y)$ that are accepted by the verifier in the Verdict stage. Note that running $\Chk$ does not require the trapdoor $sk$, but the relation implicitly depends on it.
For notational convenience, we will also denote $\Pi_{sk,\test}$ as $\Pi_{s,sk,\test}$, though the projector does not depend on $s$.

Second, the two Hadamard round checks \ref{h1-check} and \ref{h2-check} of the Verdict stage are represented by projectors
$\Lambda_{sk,\hada,1}$ and $\Lambda_{s,sk,\hada,2}$, respectively. These two projectors commute since they are both diagonal in the standard basis. We define the overall Hadamard round projector $\Pi_{s,sk,\hada}:=\Lambda_{sk,\hada,1}\Lambda_{s,sk,\hada,2}$.

\subsection{A lemma for the single-copy protocol}

We begin by showing an important fact about the single-copy protocol: the verifier's accepting paths associated to the two challenges (denoted $\test$ and $\hada$ for ``test'' and ``Hadamard,'' respectively) correspond to nearly orthogonal\footnote{Strictly speaking, the projectors are only nearly orthogonal when applied to states prepared by efficient provers.} projectors. Moreover, in a certain sense this property holds even for input states that are adaptively manipulated by a dishonest prover after they have learned which challenge will take place. This fact is essential in our analysis of the parallel repetition of many copies in the following sections.

\paragraph{The setup.}

As discussed in \cite{Mah18}, any prover $\prover$ can be characterized as follows.
First, pick a state family $\ket{\Psi_{pk}}$; this state is prepared on registers $WXYE$ after receiving $pk$.
Here $Y$ is the register that will be measured in Round $\prover_1$, $W$ and $X$ are the registers that will be measured in Round $\prover_2$, and $E$ is the private workspace of $\prover$. Then, choose two unitaries $U_\test$ and $U_\hada$ to describe the Round~$\prover_2$ actions of $\prover$ before any measurements, in the test round and Hadamard round, respectively. Both $U_\test$ and $U_\hada$ act on $WXYE$, but can only be classically controlled on $Y$, as they must be implemented after $\prover$ has measured $Y$ and sent the result to the verifier. (Of course, a cheating prover is not constrained to follow the honest protocol, but we can nevertheless designate a fixed subsystem $Y$ that carries their message.) We will write $\prover = (\ket{\Psi_{pk}}, U_\test, U_\hada)$, where it is implicit that $\ket{\Psi_{pk}}$ is a family of states parameterized by $pk$.

At the end of the protocol, the registers $WXY$ are measured and given to the verifier. Recall that we can view the final actions of the verifier as applying one of two measurements: a test-round measurement or a Hadamard-round measurement. Let $\Pi_{s,sk,\test}$ and $\Pi_{s,sk,\hada}$ denote the ``accept'' projectors for those measurements, respectively. For a given prover $\prover$, we additionally define
\begin{align*}\label{eq:single-copy-projections}\nonumber
\Pi_{s,sk,\test}^{U_\test} &:= U_\test^\dag(\Pi_{s, sk, \test} \otimes\1_E)U_\test\,, \\
\Pi_{s,sk,\hada}^{U_\hada} &:= U_\hada^\dag(H_{WX} \Pi_{s, sk, \hada} H_{WX} \otimes\1_E)U_\hada\,,
\end{align*}
where $H_{WX}$ denotes the Hadamard transform on registers $WX$, i.e., the Hadamard gate applied to every qubit in those registers. These projectors have a natural interpretation: they describe the action of the two accepting projectors of the verifier on the initial state $\ket{\Psi_{pk}}$ of the prover, taking into account the (adaptive) attacks the prover makes in Round $\prover_2$.

\paragraph{A key lemma.}

We now prove a fact about the single-copy protocol. The proof is largely a matter of making some observations about the results from \cite{Mah18}, and then combining them in the right way.

Recall that, after the setup phase, for any instance $H$ of the ZX-Hamiltonian problem (\prob{ZX}), $\mathfrak{M}$ begins with the verifier $\verifier$ making a measurement basis choice $h\in\bit^{nr}$ for all the qubits.
After interacting with a prover $\prover$, the verifier either rejects or produces a candidate measurement outcome, which is then tested as in \prot{modified-MF}.
We let $D_{\prover,h}$ denote the distribution of this candidate measurement outcome for a prover $\prover$ and basis choice $h$, averaged over all measurements and randomness of $\prover$ and $\verifier$. It is useful to compare $D_{\prover, h}$ with an ``ideal'' distribution  $D_{\rho, h}$ obtained by simply measuring some $(nr)$-qubit quantum state $\rho$ (i.e., a candidate ground state) according to the basis choices specified by $h$, with no protocol involved.
\begin{lemma}\label{lem:perfect-succ-prob}
Let $\prover = (\ket{\Psi_{pk}}, U_\test, U_\hada)$ be a prover in $\mathfrak{M}$ such that, for every $h \in \bit^{nr}$ and $s\in\bit^p$,
\begin{equation}\label{eq:perfect}
\Exp_{(pk,sk)\gets\Gen(1^\lambda,h)}[\bra{\Psi_{pk}}\Pi_{s,sk,\test}^{U_\test}\ket{\Psi_{pk}}] \geq 1 - \negl(n)\,.
\end{equation}
Then there exists an $(nr)$-qubit quantum state $\rho$ such that, for every $h,s$,
\begin{equation*}
\Exp_{(pk,sk)\gets\Gen(1^\lambda,h)}[\bra{\Psi_{pk}}\Pi_{s,sk,\hada}^{U_\hada}\ket{\Psi_{pk}}]\leq \alpha_{h,s,\rho}+\negl(n)\,,
\end{equation*}
where $\alpha_{h,s,\rho}$ (see \rmk{alpha}) is the success probability in the MF protocol with basis choice $h$ and quantum state $\rho$.
\end{lemma}
\begin{proof}
Up to negligible terms, \eq{perfect} means that $\prover$ is what Mahadev calls a \emph{perfect prover}.
She establishes two results ({\cite[Claim~7.3]{Mah18}} and {\cite[Claim~5.7]{Mah18}})
which, when taken together, directly imply the following fact about perfect provers.
For every perfect prover $\prover$, there exists an efficiently preparable quantum state $\rho$ such that $D_{\prover,h}$ is computationally indistinguishable from $D_{\rho,h}$ for all basis choices $h \in \bit^{nr}$.
In particular, the proof is obtained in two steps.
First, for every perfect prover, there exists a nearby ``trivial prover'' whose attack in a Hadamard round commutes with standard basis measurement on the committed state \cite[Claim~5.7]{Mah18}.
Second, for every trivial prover, the distribution is computationally indistinguishable from measuring a consistent quantum state $\rho$ in any basis $h$ \cite[Claim~7.3]{Mah18}.
Mahadev shows this for exactly perfect provers, but the proofs can be easily adapted  to our ``negligibly-far-from-perfect'' case.

Now consider two ways of producing a final accept/reject output of the verifier.
In the first case, an output is sampled from the distribution $D_{\prover, h}$ and the verifier applies the final checks in $\mathfrak{M}$.
In this case, the final outcome is obtained by performing the measurement $\{\Pi_{s,sk,\hada}^{U_\hada}, \1 - \Pi_{s,sk,\hada}^{U_\hada}\}$ on the state $\ket{\Psi_{pk}}$, and accepting if the first outcome is observed.
In the second case, an output is sampled from the distribution $D_{\rho, h}$ and the verifier applies the final checks in the MF protocol.
In this case, the acceptance probability is $\alpha_{h,s,\rho}$ simply by definition.
The result then follows directly.
\end{proof}

Notice that for the soundness case, there is no state that succeeds non-negligibly in the MF protocol.
In this case, \lem{perfect-succ-prob} implies that for perfect provers the averaged projection $\Exp_{(pk,sk)\gets\Gen(1^\lambda,h),h,s}[\bra{\Psi_{pk}}\Pi_{s,sk,\hada}^{U_\hada}\ket{\Psi_{pk}}]$ is negligible. In other words,
provers who succeed almost perfectly in the test round must almost certainly fail in the Hadamard round. We emphasize that this is the case even though the prover can adaptively change their state (by applying $U_\test$ or $U_\hada$) after learning which round will take place. This formalizes the intuitive claim we made at the beginning of the section about ``adaptive orthogonality'' of the two acceptance projectors corresponding to the two round types.

\subsection{The parallel repetition theorem}

\paragraph{Characterization of a prover in the $k$-fold protocol.}

We now discuss the behavior of a general prover in a $k$-fold protocol. We redefine some notation, and let $\verifier$ be the verifier and $\prover$ an arbitrary prover in the $k$-fold protocol.

In the Setup phase, the key pairs $(pk_1,sk_1),\ldots,(pk_k,sk_k)$ are sampled according to the correct NTCF/NTIF distribution.\footnote{Recall that the keys are sampled by choosing uniform bases $h$ followed by running $\Gen(1^\lambda,h)$.}
The secret keys $sk = (sk_1,\ldots,sk_k)$ and the corresponding bases $h$ are given to $\verifier$, whereas $pk = (pk_1,\ldots,pk_k)$ is given to $\prover$ (in the register $PK =(PK_1,\ldots,PK_k)$).

In Round~$\prover_1$, without loss of generality, the action of $\prover$ prior to measurement is to apply a unitary $U_0=\sum_{pk}\proj{pk}_{PK}\otimes (U_{0,pk})_{WXYE}$ to the input state $\ket{pk}_{PK}\ket{0}_{WXYE}$. Each of $W,X,Y$ is now a $k$-tuple of registers, and $E$ is the prover's workspace. To generate the ``commitment'' message to $\verifier$, $\prover$ performs standard basis measurement on $Y$.
We write $\ket{\Psi_{pk}}_{WXYE}=\sum_y\beta_y\ket{\Psi_{pk, y}}_{WXE}\ket{y}_Y$. When the measurement outcome is $y$, the side state $\prover$ holds is then $\ket{\Psi_{pk,y}}_{WXE}$. In the following analysis of the success probability of $\prover$,
we consider the superposition $\ket{\Psi_{pk}}_{WXYE}$ instead of a classical mixture of the states $\ket{\Psi_{pk,y}}_{WXE}$
using the principle of deferred measurement.

In Round $\verifier_2$, after receiving the commitment $y=(y_1,\ldots,y_k)$ from the prover, $\verifier$ sends challenge coins $c := (c_1,\ldots,c_k)\in\{0,1\}^k$. For the remainder of the protocol, we take the following point of view. We assume that $\prover$ and $\verifier$ share access to a register $Y$ whose state is fixed forever to be the standard basis state $y$. This is clearly equivalent to the real situation, as there $\prover$ measured $Y$ and committed to the outcome $y$ by sending it to $\verifier$.

In Round $\prover_2$, without loss of generality, the action of $\prover$ consists of a general operation (that can depend on $c$), followed by the honest action. The general operation is some efficient unitary $U_c$ on $WXYE$. The honest action is measurement in the right basis, i.e., for each $i$, $W_iX_i$ is measured in the standard basis (if $c_i=0$) or the Hadamard basis (if $c_i=1$). Equivalently, the honest action is \emph{(i.)} apply $\mathfrak H^c_{WX}:=\bigotimes_{i=1}^k (H^{c_i})_{W_iX_i}$, i.e., for each $\{i : c_i = 1\}$ apply a Hadamard to every qubit of $W_iX_i$, and then \emph{(ii.)} apply standard basis measurement.

In the Verdict stage, $\verifier$ first applies for each $i$ the two-outcome measurement corresponding to the $\Pi_{s_i,sk_i,c_i}$ from the single-copy protocol. The overall decision is then to accept if the measurements accept for all $i$. We let
\begin{equation}\label{eq:prod-meas}
\left(\Pi_{s,sk,c}\right)_{WXY}:=\bigotimes_{i=1}^k\left(\Pi_{s_i,sk_i,c_i}\right)_{W_iX_iY_i}
\end{equation}
denote the corresponding acceptance projector for the entire $k$-copy protocol. The effective measurement on $\ket{\Psi_{pk}}_{WXYE}$ is then described by the projection
\begin{align*}
\left(\Pi_{s,sk,c}^{U_c}\right)_{WXYE}:=(U_c^\dag)_{WXYE}(\mathfrak H^c\Pi_{s,sk,c,y} \mathfrak H^c\otimes\1_E)(U_c)_{WXYE}\,.
\end{align*}
The success probability of $\prover$, which is characterized by the state $\ket{\Psi_{pk}}$ and family of unitaries $\{U_c\}_{c \in \bit^n}$, is thus
\begin{align*}
\Exp_{(pk,sk)\gets\Gen(1^\lambda,h),h,s,c}\bigl[\bra{\Psi_{pk}}\Pi_{s,sk,c}^{U_c}\ket{\Psi_{pk}}\bigr].
\end{align*}

\paragraph{The proof of parallel repetition.}\label{sec:proof}

Recall that \lem{perfect-succ-prob} states that the projectors corresponding to the two challenges in $\mathfrak{M}$ are nearly orthogonal, even when one takes into account the prover's adaptively applied unitaries. We show that this property persists in the $k$-copy protocol. Specifically, we show that all $2^k$ challenges are nearly orthogonal (in the same sense as in \lem{perfect-succ-prob}) with respect to any state $\ket{\Psi_{pk}}$ and any post-commitment unitaries $U_c$ of the prover.

This can be explained informally as follows. For any two distinct challenges $c \neq c'$, there exists a coordinate $i$ such that $c_i \neq c_i'$, meaning that one enters a test round in that coordinate while the other enters a Hadamard round. In coordinate $i$, by the single-copy result (\lem{perfect-succ-prob}), the prover who succeeds with one challenge should fail with the other. A complication is that, since we are dealing with an interactive argument, we must show that a violation of this claim leads to an \emph{efficient} single-copy prover that violates the single-copy result. Once we have shown this, we can then apply it to any distinct challenge pairs $c \neq c'$. It then follows that we may (approximately) decompose $\ket{\Psi_{pk}}$ into components accepted in each challenge, each of which occurs with probability $2^{-k}$. We can then use this decomposition to express the overall success probability of $\prover$ in terms of this decomposition. As $\ket{\Psi_{pk}}$ is of course a normalized state, it will follow that the overall soundness error is negligibly close to $2^{-k}$.

The ``adaptive orthogonality'' discussed above is formalized in the following lemma. Recall that any prover in the $k$-fold parallel repetition of $\mathfrak{M}$ can be characterized by a state family $\{\ket{\Psi_{pk}}\}_{pk}$ that is prepared in Round~$\prover_1$ and a family of unitaries $\{U_c\}_{c\in\{0,1\}^k}$ that are applied in Round~$\prover_2$.

\begin{restatable}{lemma}{ortm}\label{lem:nearly-orthogonal-measurements}
Let $\prover$ be a prover in the $k$-fold parallel repetition of $\mathfrak{M}$
that prepares $\ket{\Psi_{pk}}$ in Round~$\prover_1$ and performs $U_c$ in Round~$\prover_2$. Let $a, b \in \bit^k$ such that $a \neq b$ and choose $i$ such that $a_i \neq b_i$. Then there is an $(nr)$-qubit quantum state $\rho$ such that for every basis choice $h$ and randomness $s$,
\begin{align*}
\Exp_{(pk,sk)\gets\Gen(1^\lambda,h)}\left[\bra{\Psi_{pk}}\Pi_{s,sk,b}^{U_b}\Pi_{s,sk,a}^{U_{a}}+\Pi_{s,sk,a}^{U_{a}}\Pi_{s,sk,b}^{U_b}\ket{\Psi_{pk}}\right] \leq 2\alpha_{h_i,s_i,\rho}^{1/2}+\negl(n)\,,
\end{align*}
where $\alpha_{h_i,s_i,\rho}$ (see \rmk{alpha}) is the success probability with $\rho$ conditioned on the event that $h_i$ is sampled.
\end{restatable}
\begin{proof}
  See \app{lemma-orthogonal-measurements}.
\end{proof}
\noindent
We remark that this adaptive orthogonality is guaranteed under a computational assumption.
Assuming that no efficient quantum adversary can break the underlying security properties based on plain $\LWE$,
the projections are pairwise orthogonal in the sense of averaging over the key pairs $(pk,sk)$  and with respect to any quantum state $\ket{\Psi_{pk}}$ prepared by an efficient quantum circuit.

We also emphasize that, in \lem{nearly-orthogonal-measurements}, for each pair $a \neq b$ the left-hand side is upper-bounded by the acceptance probability of measuring some state $\rho$ in the basis $h_i$, and the quantum state $\rho$ may be different among distinct choices of $(a,b)$ and $i$.
This implies that if $\prover$ succeeds with one particular challenge perfectly\footnote{More concretely, if  for some fixed $a$, $\Pi_{s,sk,a}\ket{\Psi_{pk}}=\ket{\Psi_{pk}}$.} when we average over  $h$ and $s$, \lem{nearly-orthogonal-measurements} and standard amplification techniques (see \sec{indep}) imply that it succeeds on challenge $b \ne a$ with probability at most
$\Exp_{(pk,sk)\gets\Gen(1^\lambda)}\bra{\Psi_{pk}}\Pi_{s,sk,b}\ket{\Psi_{pk}}\leq \negl(n)$.
We note that this strategy leads to acceptance probability at most $2^{-k}+\negl(n)$.

Unfortunately, the above observation  does not rule out that $\prover$ can succeed on several challenges with appreciable probability, achieving a better overall success probability. For instance, one could try to implement this by coherently simulating distinct provers who perfectly win different challenges.

We rule out this possibility  by showing that if the projectors are pairwise nearly orthogonal with respect to a quantum state $\rho$,
then the above strategy is optimal with respect to $\rho$. Since pairwise orthogonality holds with respect to \emph{any} efficiently preparable quantum state by \lem{nearly-orthogonal-measurements}, our parallel repetition theorem follows.

First, we state a key technical lemma, proved in \app{lemmas}.

\begin{restatable}{lemma}{ort}\label{lem:orthogonality-implies-soundness}
  Let $A_1,\ldots,A_m$ be projectors and $\ket{\psi}$ be a quantum state. Suppose there are real numbers $\delta_{ij}\in[0,2]$ such that $\bra{\psi}A_iA_j+A_jA_i\ket{\psi}\leq \delta_{ij}$ for all $i \ne j$. Then
$\bra{\psi}A_1+\cdots+A_m\ket{\psi}\leq 1+\bigl(\sum_{i<j}\delta_{ij}\bigr)^{1/2}$.
\end{restatable}
\noindent
Observe that when the projectors are mutually orthogonal, we have $A_1 + \cdots + A_m \preceq \1$ and the bound clearly holds.
\lem{orthogonality-implies-soundness} describes a relaxed version of this fact.
In our application, the projectors and the state are parameterized by the key pair, and we use this bound to show that the average of pairwise overlaps is small.
We are now ready to establish our parallel repetition theorem.

\begin{lemma}\label{lem:soundness-k}
  Let $k$ be a positive integer and let $\{U_c\}_{c\in\{0,1\}^k}$ be any set of unitaries that may be implemented by $\prover$ after the challenge coins are sent.
  Let $\ket{\Psi_{pk}}$ be any state $\prover$ holds in the commitment round, and suppose $\prover$ applies $U_c$ followed by honest measurements when the coins are $c$.
  Then there exists a negligible function $\epsilon$ such that $\verifier_1,\ldots,\verifier_k$ accept $\prover$ with probability at most $2^{-k}+\epsilon(n)$.
\end{lemma}

\begin{proof}
  The success probability of any prover in the $k$-fold protocol is
  \begin{align*}
    \Pr[\text{success}]
    &= 2^{-k}\Exp_{(pk,sk)\gets\Gen(1^\lambda,h),h,s}[\bra{\Psi_{pk}}\sum_c\Pi_{s,sk,c}^{U_c}\ket{\Psi_{pk}}]
  \end{align*}
  where $h,s$ are drawn from uniform distributions.

  Define a total ordering on $\{0,1\}^k$ such that $a<b$ if $a_i<b_i$ for the smallest index $i$ such that $a_i\neq b_i$.
  Then by \lem{orthogonality-implies-soundness}, we have
  \begin{align*}
    \Pr[\text{success}]
    &\leq 2^{-k} + 2^{-k}\Exp_{h,s}\left[\sum_{a<b} \Exp_{(pk,sk)\gets\Gen(1^\lambda,h)}[\bra{\Psi_{pk}}\Pi_{s,sk,a}^{U_a}\Pi_{s,sk,b}^{U_b}+\Pi_{s,sk,b}^{U_b}\Pi_{s,sk,a}^{U_a}\ket{\Psi_{pk}}]\right]^{1/2}.
  \end{align*}
  By \lem{nearly-orthogonal-measurements},
  there exists a negligible function $\delta$ such that
\begin{align*}
\Exp_{(pk,sk)\gets\Gen(1^\lambda,h)}[\bra{\Psi_{pk}}\Pi_{s,sk,a}^{U_a}\Pi_{s,sk,b}^{U_b}+\Pi_{s,sk,b}^{U_b}\Pi_{s,sk,a}^{U_a}\ket{\Psi_{pk}}]\leq 2\alpha_{h_{i(a,b)},\rho_{ab}}^{1/2}+\delta\,
\end{align*}
for every pair $(a,b)$.
Here $i(a,b)$ is the smallest index $i$ such that $a_i\neq b_i$ and $\rho_{ab}$ is the reduced quantum state associated with $a,b$, as guaranteed by \lem{nearly-orthogonal-measurements}.
Let $\mu$ be the soundness error of the MF protocol. We have
  \begin{align*}\nonumber
    \Pr[\text{success}]
    &\leq 2^{-k} + 2^{-k} \Exp_{h,s}\left[\sum_{a<b} \left(2\alpha_{h_{i(a,b)},s_{i(a,b)},\rho_{ab}}^{1/2}+\delta\right) \right]^{1/2} \\\nonumber
    &\leq 2^{-k} + 2^{-k} \Exp_{h,s}\left[\sum_{a<b} 2\alpha_{h_{i(a,b)},s_{i(a,b)},\rho_{ab}}^{1/2}\right]^{1/2} + 2^{-k}\sqrt{\binom{2^k}{2}}\delta^{1/2} \\\nonumber
    &\leq 2^{-k} + 2^{-k} \left[\sum_{a<b} 2\left(\Exp_{h,s}[\alpha_{h_{i(a,b)},s_{i(a,b)},\rho_{ab}}]\right)^{1/2}\right]^{1/2} + \delta^{1/2} \\\nonumber
    &\leq 2^{-k} + 2^{-k} \left[\sum_{a<b} 2\mu^{1/2}\right]^{1/2} + \delta^{1/2} \\\nonumber
    &\leq 2^{-k} + 2^{-k} \left[2^k(2^k-1)\mu^{1/2}\right]^{1/2} + \delta^{1/2} \\
    &\leq 2^{-k} + \mu^{1/4} + \delta^{1/2}
  \end{align*}
  where the second
  and third inequalities hold by Jensen's inequality.
  Amplifying the soundness of the MF protocol, $\mu$ is negligible using polynomially many copies by \lem{modified-mf}.
  Thus the soundness error is negligibly close to $2^{-k}$.
\end{proof}

We note that Mahadev shows the soundness error for a single-copy protocol is negligibly close to $3/4$ \cite{Mah18}, whereas \lem{soundness-k} implies the error can be upper bounded by $1/2+\negl(n)$.
Mahadev obtains soundness error $3/4+\negl(n)$ by considering a general prover $\prover$ who, for each basis $h$, succeeds in the test round (characterized by $\Pi_{h,sk,\test}$) with probability $1-p_{h,\test}$, in the first stage of the Hadamard round (characterized by $\Lambda_{h,sk,\hada,1}$) with probability $1-p_{h,\hada}$, and in the second stage of the Hadamard round (characterized by $\Lambda_{h,sk,\hada,2}$) with probability at most $\sqrt{p_{h,\test}}+p_{h,\hada}+\alpha_{h,\rho}+\negl(n)$ for some state $\rho$~\cite[Claim~7.1]{Mah18}. These contributions are combined by applying the triangle inequality for trace distance.
This analysis is loose since $\Lambda_{h,sk,\hada,1}$ and $\Lambda_{h,sk,\hada,2}$ commute, and $\prover$ must pass both stages to win the Hadamard round.

Finally, \lem{soundness-k} immediately implies the following theorem.

\begin{theorem}\label{thm:parallel-soundness}
Let $\mathfrak{M}^k$ be the $k$-fold parallel repetition of the three-round Mahadev protocol $\mathfrak{M}$. Under the \LWE assumption, the soundness error of $\mathfrak{M}^k$ is at most $2^{-k} + \negl(n)$.
\end{theorem}

We conclude with the following three-round protocol.

\begin{protocol}[Verification with instance-independent setup]\label{prot:interactive-attempt}~
\begin{proto}
  \item[Setup.]
    $\verifier$ samples random bases $h\in\bit^{nrk}$ and runs the key generation algorithm $(pk,sk)\gets\Gen(1^\lambda,h)$.
    $\verifier$ samples a string $s\in\bit^{prk}$ uniformly.
    $\verifier$ sends the public keys $pk$ to $\prover$.
  \item[Round~$\prover_1$.] $\prover$ queries $\Samp$ coherently on the witness state $\ket{\psi}^{\otimes rk}$, followed by a standard basis measurement on register $Y$.
    The outcome is sent to $\verifier$.

  \item[Round~$\verifier_2$.]
    $\verifier$ samples $c_1,\ldots,c_k\gets\{0,1\}$ and sends $c=(c_1,\ldots,c_k)$ to $\prover$.
  \item[Round~$\prover_2$.]
    For each $i\in[k]$, $j\in[r]$, $\ell\in[n]$,
    \begin{enumerate}
      \item
        if $c_i=0$, $\prover$ performs a standard basis measurement and gets $u_{ij\ell}=(w_{ij\ell},t_{ij\ell})$;
      \item
        if $c_i=1$, $\prover$ performs a Hadamard basis measurment and gets $u_{ij\ell}=(w_{ij\ell},t_{ij\ell})$.
    \end{enumerate}
    $\prover$ sends $u$ to $\verifier$.

  \item[Verdict.]
    For each $i\in[k]$,
    \begin{enumerate}
      \item If $c_i=0$,
        $\verifier$ accepts iff $\bigwedge_{j,\ell}\Chk(pk_{j\ell},w_{j\ell},t_{j\ell},y_{j\ell})=1$.
      \item
        If $c_i=1$, $\verifier$ records the set $A_i\subseteq [r]$ of consistent copies.
        For each $j\in A_i$ and $\ell\in[n]$:
        \begin{enumerate}
          \item If $h_{ij\ell}=0$, run $(b_{ij},x_{b_{ij},y_{ij}})\gets\Inv_\G(sk_{ij},y_{ij})$. Set $e_{ij\ell}=b_{ij\ell}$;
            if $h_{ij}=1$, compute $x_{0,y_{ij\ell}},x_{1,y_{ij\ell}}$ and $e_{ij\ell}=t_{ij\ell}\cdot(x_{0,y_{ij\ell}}\oplus x_{1,y_{ij\ell}})\oplus w_{ij}$.
            If any of the algorithms rejects or any of $t_{ij\ell}$ is trivial (e.g., $t_{ij\ell}=0$, see \cite{Mah18}), $\verifier$ sets $v_{ij}=0$; otherwise enters the next step.
          \item
            $\verifier$ computes the terms $S_{ij}=\term(H,s_{ij})$ for each $i\in[k],j\in[r]$.
            Set $v_{ij}=1$ if $(e_{ij\ell})_{\ell\in S_{ij}}$ satisfies $M_{-m_{S_{ij}}}$ and $v_{ij}=0$ otherwise.
        \end{enumerate}
        Then $\verifier$ sets $v_i=1$ if $\sum_{j\in A_i} v_{ij}\geq (2-a-b)|A_i|/4$ and 0 otherwise.
    \end{enumerate}
    $\verifier$ accepts iff $v_i=1$ for every $i\in[k]$. \\
    The verdict function is $\verdict(H,s,sk,y,c,u) := \bigwedge_{i=1}^k v_i$.
\end{proto}
\end{protocol}
\begin{theorem}\label{thm:interactive}
  For $r=\omega(\frac{\log n}{(b-a)^2})$ and $k=\omega(\log n)$,
  \prot{interactive-attempt} is a quantum prover interactive argument for $\zx_{a,b}$ with negligible completeness error and soundness error. \end{theorem}

\section{A classical zero-knowledge argument for $\QMA$}\label{sec:zk}

To turn \prot{interactive-attempt} into a zero-knowledge protocol, we first consider an intermediate protocol in which $\prover$ first encrypts the witness state $\ket{\psi}^{\otimes rk}$ with a quantum one-time pad.
Then in Round $\prover_2$, $\prover$ sends the one-time pad key $\beta,\gamma$ along with the response $u$.
In the verdict stage, $\verifier$ uses the keys to decrypt the response.
We denote the verdict function as
\begin{align}\label{eq:verdict}
\verdict'(H,s,sk,y,c,\beta,\gamma,u):=\verdict(H_{\beta,\gamma}, s, sk, y, c, u)
\end{align}
where $H_{\beta,\gamma}:= X^\beta Z^\gamma H Z^\gamma X^\beta$ is the instance conjugated by the one-time pad.

Obviously, this is not zero-knowledge yet, as the verifier can easily recover the original measurement outcomes on the witness state. To address this issue, we take the approach of \cite{BJSW16,CVZ19} and invoke a NIZK protocol for \NP languages. The language $\LL$ that we consider is defined as follows:
\begin{align*}
  \LL:=\{&(H,s,sk,\xi,y,c,\chi):~\exists~\witness=(\beta,\gamma,u,r_1,r_2),\\
  & \xi = \commit(u;r_1)
  \wedge \chi = \commit(\beta,\gamma; r_2) \\
  & \wedge \verdict'(H,s,sk,y,c,\beta,\gamma,u)=1
  \},
\end{align*}
where $r_1,r_2$ are the randomness for a computationally secure bit commitment scheme (see \app{standard-primitives}).
However, this alone is insufficient since, to agree on an instance without introducing more message exchanges, $\verifier$ must reveal $sk,s$ before $\prover$ sends a NIZK proof. Revealing $sk,s$ enables a simple attack on soundness:
$\prover$ can ensure the verifier accepts all instances by using the secret key to forge a valid response $u$, committing to it, and computing the NIZK proof.

The solution is to invoke a quantum-secure classical FHE scheme and to let $\prover$ homomorphically compute a NIZK proof.
This requires $\prover$ to only use an encrypted instance.
In the setup phase, $\prover$ is given the ciphertexts $csk=\FHE.\Enc_{hpk}(sk)$ and $cs=\FHE.\Enc_{hpk}(s)$.
Next, in Round~$\prover_2$, $\prover$ computes
$cx=\FHE.\Enc_{hpk}(x)$ where $x:=(H,s,sk,\xi,y,c,\chi)$
since the partial transcript $(y,c,\xi,\chi)$ has already been fixed.
$\prover$ then computes
\begin{align*}
ce=\FHE.\Eval_{hpk}(\NIZK.\aP,cc,cx,c\tau)= \FHE.\Enc_{hpk}(\NIZK.\aP(\crs,x,\witness)),
\end{align*}
where $c\tau=\FHE.\Enc_{hpk}(\witness)$,
and sends $ce$ to $\verifier$.
Finally, $\verifier$ decrypts the encrypted NIZK proof $ce$ and outputs $\NIZK.\aV(\crs,x,e)$.
The above transformation yields a three-message zero-knowledge protocol for quantum computation with trusted setup from a third party, described as follows.

\begin{protocol}[Setup phase $\setup(\lambda,N,M)$]\label{prot:setup}
  The algorithm $\setup$ takes two integers $N,M$ as input, and outputs two strings $\st_\verifier,\st_\prover$ with the following steps.
  \begin{enumerate}
    \item Run $\crs\gets\NIZK.\Setup(1^\lambda)$.
    \item Sample uniform bases $h\gets\bit^N$ and run $(pk,sk)\gets\Gen(1^\lambda,h)$.
    \item Run the FHE key generation algorithm $(hpk,hsk)\gets\FHE.\Gen(1^\lambda)$.
    \item Run encryption on the secret key $csk\gets \FHE.\Enc_{hpk}(sk)$.
    \item Choose keys $(\beta,\gamma)$ and randomness $r_1$ uniformly and compute $\xi=\commit(\beta,\gamma;r_1)$.
    \item Sample a random string $s_1,\ldots,s_M\in\bit^{p}$ (see \rmk{alpha}) uniformly and compute its encryption $cs=\FHE.\Enc_{hpk}(s)$.
  \end{enumerate}
  Output $\st_\verifier=(\crs,sk,hsk,hpk,\xi)$ and $\st_\prover=(\crs,pk,hpk,csk,cs,\beta,\gamma,r_1)$.
\end{protocol}
\begin{protocol}[An interactive protocol]\label{prot:interactive}~
\begin{proto}
  \item[Setup.]
    Run $\st_\verifier,\st_\prover\gets\setup(\lambda,nrk,rk)$.
    Send $\st_\verifier=(\crs,sk,hsk,hpk,\xi)$ to $\verifier$ and $\st_\prover=(\crs,pk,hpk,csk,cs,\beta,\gamma,r_1)$ to $\prover$.
  \item[Round~$\prover_1$.]
    $\prover$ aborts if $pk$ is invalid.
    $\prover$ queries $\Samp$ coherently on the witness state $X^\beta Z^\gamma\ket{\psi}^{\otimes rk}$.

  \item[Round~$\verifier_2$.]
    $\verifier$ samples $c_1,\ldots,c_k\gets\{0,1\}$ and sends $c=(c_1,\ldots,c_k)$ to $\prover$.
  \item[Round~$\prover_2$.]
    For each $i\in[k]$, $j\in[r]$, $\ell\in[n]$,
    \begin{enumerate}
      \item
        if $c_i=0$, $\prover$ performs a standard basis measurement and gets $u_{ij\ell}=(w_{ij\ell},t_{ij\ell})$.
      \item
        if $c_i=1$, $\prover$ performs a Hadamard basis measurement and gets $u_{ij\ell}=(w_{ij\ell},t_{ij\ell})$.
    \end{enumerate}
    $\prover$ sends $\chi:=\commit(u; r_2)$ and
    \begin{align*}
      ce:=\FHE.\Eval_{hpk}(\NIZK.\aP, cc,cx,c\tau),
    \end{align*}
    where $cc$, $cx$ and $c\tau$ are the encryptions of $\crs$, $x$ and $\tau$ respectively.

  \item[Verdict.]
    $\verifier$ accepts if
    $\NIZK.\aV(\crs,x,\FHE.\Dec_{hsk}(ce))=1$.
\end{proto}
\end{protocol}

\subsection{Completeness and soundness}

\begin{theorem}\label{thm:soundness-zk}
  \prot{interactive} has  negligible  completeness and soundness errors.
\end{theorem}
\begin{proof}
  The completeness follows from the correctness of $\FHE$, the completeness of $\NIZK$, and the completeness of \prot{interactive-attempt}.

To prove soundness, we consider the following hybrid protocols:
\hybrids

 We show in \app{soundness-zk} that for each $H_i$, the maximum acceptance probability of any bounded prover is at most negligible for no instances, proving soundness.
\end{proof}

\subsection{The zero-knowledge property}
\label{sec:zk-property}
To show \prot{interactive} is zero-knowledge, we consider the following simulator. \\[1em]
\textbf{Simulator $\simulator(H)^{\verifier_2^*}$.}
  \begin{proto}
  \item[Setup.] $\simulator$ samples pair $\st_\verifier,\st_\prover$ according to the correct distributions.
    \item[$\prover_1$.] For each bit, $\simulator$ samples random $b_{ij\ell},x_{ij\ell}$ and queries the function $f_{pk_{ij\ell}}$ to get a sample $y_{ij\ell}$.
      $\simulator$ computes $\xi=\commit(\beta,\gamma;r_1)$ for random $\beta$ and $\gamma$ and feeds $\verifier^*$ with $y$.
    \item[$\verifier_2$.] $\simulator$ queries $\verifier_2^*$ to get coins $c=(c_1,\ldots,c_k)$.
    \item[$\prover_2$.]
      $\simulator$ uses $\st_\verifier$ to compute a valid response $u$.\footnote{The response $u$ is computed as follows: for each $i$ such that $c_i=0$, $\simulator$ sets $u_{ij\ell}=(b_{ij\ell},x_{ij\ell})$.
        For each $i$ such that $c_i=1$, $\simulator$ computes the terms $S_{ij}=\term(H,s_{ij})$. For each copy $i,j$ such that $S_{ij}$ is consistent with $h_{ij}$, do the following:
        $\simulator$ samples nontrivial $t_{ij\ell}$ uniformly and uses the secret key $sk_{ij\ell}$ to compute the hardcore bit $o_{ij\ell}=t_{ij\ell}\cdot(x_{0y_{ij\ell}}\oplus x_{1y_{ij\ell}})$.
        Then for each copy $i,j$ and qubits $\ell_1,\ell_2\in S_{ij}$ such that $S_{ij\ell_1}=X$ (and $S_{ij\ell_2}=Z$), set $b_{ij\ell_2}'$ randomly and $b_{ij\ell_1}'= \beta_{ij\ell}\oplus b_{ij\ell_2}\oplus\frac{1-m_{S_{ij}}}{2}$; for $\ell\notin S_{ij}$, samples $b_{ij\ell}'$ uniformly. The response $u_{ij\ell}=(b_{ij\ell}'\oplus o_{ij\ell}\oplus\gamma_{ij\ell} ,t_{ij\ell})$.
      For inconsistent copies, $\simulator$ samples a random response. }
      $\simulator$ computes the commitment $\chi=\commit(u;r_2)$ and $ce=\FHE.\Enc_{hpk}(\NIZK.\aS(x))$.
  \end{proto}
  Finally, $\simulator$ outputs $(\st_\verifier,\xi,y,c,\chi,ce)$.

To prove zero-knowledge, we consider an intermediate (quantum) simulator $\simulator_Q$ that simulates the interaction between $\prover$ and $\verifier^*$, but that sets the last message as
$ce=\FHE.\Enc_{hpk}(\NIZK.\aS(x))$ instead of the $\NIZK$ proof.

\begin{lemma}\label{lem:zk-1}
  The output distributions of the original protocol and $\simulator_Q$ are computationally indistinguishable.
\end{lemma}
\begin{proof}
  We prove this by showing an intermediate quantum simulator as follows:
  \begin{enumerate}
    \item $\tilde{\simulator}_0$: A quantum simulator that fully simulates $\prover$ and $\verifier$ in the original protocol.
    \item $\tilde{\simulator}_{1}$:
      Same as $\tilde{\simulator}_{0}$, except that $\tilde{\simulator}_{1}$ sets $ce=\FHE.\Enc_{hpk}(\NIZK.\aP(\crs,x,\witness))$, instead of computing it homomorphically  (this is possible since the simulator has access to the private keys).
    \item $\tilde{\simulator}_{2}$: Same as $\tilde{\simulator}_{1}$, except that $\tilde{\simulator}_2$ sets $ce=\FHE.\Enc_{hpk}(\NIZK.\aS(x))$. Notice that this is $\simulator_Q$.
  \end{enumerate}

  The output distribution of $\tilde{\simulator}_0$ is indistinguishable from the output of $\tilde{\simulator}_1$ by the circuit privacy property of the FHE scheme and that the encryption of $sk$ and $s$ are provided by a trusted party.

  The output distribution of $\tilde{\simulator}_1$ is indistinguishable from the output of $\tilde{\simulator}_2$ by the computational zero-knowledge property of $\NIZK$.
\end{proof}

\begin{lemma}\label{lem:zk-2}
  The output distributions of $\simulator_Q$ and $\simulator$ are computationally indistinguishable.
\end{lemma}
\begin{proof}
  The main difference between the protocols is that the instance $x$ for the \NP language $\LL$ is generated by the quantum simulator $\simulator_Q$, who honestly measures the quantum witness.
  Here again we introduce an intermediate step $\tilde{\simulator}_Q$ which is the same as $\simulator_Q$ but the commitment
   $\xi'=\commit(\beta',\gamma';r_1)$  is provided instead of the commitment of the one-time pad key $\beta,\gamma$, for uniformly independently random $\beta'$ and $\gamma'$.

  Notice that the output distribution of $\simulator_Q$ is computationally indistinguishable from the output of $\tilde{\simulator}_{Q}$ from the hiding property of the commitment scheme, and we just need to argue now that the output of $\tilde{\simulator}_{Q}$ is computationally indistinguishable from the output of $\simulator$.

  Since the last message is obtained by homomorphic evaluation of $\NIZK.\aS$ on the encryption of $\crs$ and $x$, it suffices to show the instances from the simulators are computationally indistinguishable.

  Notice that in $\tilde{\simulator}_Q$, $\xi'$ is independent of the other parts of the message, and in particular there is no information leaked from $\beta$ and $\gamma$ that were used to one-time pad the state. In this case, the distribution
  of the output $\Samp$ on the maximally mixed state, i.e.,
  for each qubit $(i,j,\ell)\in[k]\times[r]\times[n]$, $y_{ij\ell}$ is obtained by measuring $\Samp(pk,\1/2)$, is negligibly close to
  \begin{align*}
    \frac{1}{2|\X|}\sum_{b\in\bit}\sum_{x,x'\in\X} \proj{b}_{W_{ij\ell}}\otimes \ket{x}\!\bra{x'}_{X_{ij\ell}} \otimes \ket{\psi_{f_{pk}(b,x)}}\!\bra{\psi_{f_{pk}(b,x')}}_{Y_{ij\ell}}.
  \end{align*}
  Measuring register $Y_{ij\ell}$, by the injective pair property, there exists a unique $x$ such that $y$ lies in the support of $f_{pk}(b,x)$, and thus upon measurement the state collapses to
  \begin{align*}
    \frac{1}{2|\X|}\sum_{b,x}\sum_{y} f_{pk}(b,x)(y) \proj{b}_{W_{ij\ell}}\otimes \ket{x}\!\bra{x}_{X_{ij\ell}}\otimes \proj{y}_{Y_{ij\ell}},
  \end{align*}
  which is identical to the state obtained by querying $f_{pk}$ on uniform $b,x$.
  This implies that the marginal distributions of $y$ on the two experiments are statistically close.
  By the hiding property of the commitment scheme, the distributions over $\xi$ are computationally indistinguishable.\footnote{Note that $y$ of $\simulator$ is independent of $\beta,\gamma$, whereas in \prot{interactive}, $y,(\beta,\gamma)$ may not be independent; for example, consider the special case where the witness is $\ket{\psi}=\ket{0^n}$.}

  In Round $\verifier_2$, since the distribution over $c$ depends on $\st_\verifier,\xi,y$,
  the distributions over $(\st_\verifier,\xi,y,c)$ are computationally indistinguishable.

  In Round $\prover_2$, by the hiding property, conditioned on $(\st_\verifier,\xi,y,c)$, the distributions over $\chi$ are computationally indistinguishable.\footnote{Namely, the commitment to a response sampled from $\prover$'s witness state is computationally indistinguishable from a deterministically computed response by $\simulator$.}
\end{proof}
\noindent
From \lem{zk-1} and \lem{zk-2}, we conclude the following theorem.
\begin{theorem}
  Assuming the existence of a non-interactive bit commitment scheme with perfect binding and computational hiding, \prot{interactive} is zero-knowledge.
\end{theorem}

\section{Round reduction by Fiat-Shamir transformation}\label{sec:FS}

In this section we show that the Fiat-Shamir transformation can be used make the $k$-fold parallel three-round Mahadev protocol $\mathfrak{M}$ non-interactive with a setup phase, while keeping both the completeness and the soundness errors negligible. This will also be the case for the zero-knowledge variant of the same, i.e., \prot{interactive}.

\subsection{Fiat-Shamir for $\Sigma$-protocols in the QROM}

The Fiat-Shamir (FS) transformation turns any public-coin three-message interactive argument, also called a $\Sigma$-protocol,
into a single-message protocol in the random oracle model (ROM). A $\Sigma$-protocol consists of the following type of interaction between $\verifier$ and $\prover$:

\begin{protocol}[$\Sigma$-protocol for a language $L$]\label{prot:sigma}
  $\verifier$ and $\prover$ receive an input~$x$.
\begin{proto}
\item[Round~1.]
  $\prover$ sends a message $y$ (called the \emph{commitment}).
\item[Round~2.]
  $\verifier$ samples a random $c$ (called the \emph{challenge}) uniformly from a finite set $\C$.
\item[Round~3.]
  $\prover$ responds with a message $m$ (called the \emph{response}).
\item[Verdict.]
  $\verifier$ outputs $V(x,y,c,m)$, where $V$ is some Boolean function.
\end{proto}
\end{protocol}

In the Fiat-Shamir transformation, the prover generates their own challenge by making a query to hash function $\H$, and then computing their response $m$ as usual. The transformed protocol thus does not require $\verifier$ to send any messages.

\begin{protocol}[FS-transformed protocol for $L$]\label{prot:FS-sigma}
$\verifier_{FS}$ and $\prover_{FS}$ receive an input $x$ and are given access to a random oracle $\H$.
\begin{proto}
\item[Round~1.]
  $\prover_{FS}$ sends a message $(y,m)$.
\item[Verdict.]
  $\verifier_{FS}$ outputs $V(x,y,\H(x,y),m)$.
\end{proto}
\end{protocol}

In the standard approach, one proves that the Fiat-Shamir transformation preserves soundness in the ROM. In this idealized cryptographic model, all parties receive oracle access to a uniformly random function $\mathcal H$. Against quantum adversaries, there is a well-known complication: a quantum computer can easily evaluate any actual instantiation of $\mathcal H$ (with a concrete public classical function) in superposition via
$$
U_\mathcal H \colon \ket{x,y}\ket{z} \mapsto \ket{x,y}\ket{z \oplus \mathcal H(x,y)}\,.
$$
We thus work in the Quantum Random Oracle Model (QROM), in which all parties receive quantum oracle access to $U_\mathcal H$.

We will make use of the following theorem of~\cite{DFMS19}; we describe the underlying reduction in \app{QROM}.
\begin{theorem}[Quantum security of Fiat-Shamir~{\cite[Theorem~2]{DFMS19}}]\label{thm:FS-sigma}
For every QPT prover $\A^\H$ in \prot{FS-sigma}, there exists a QPT prover $\mathcal S$ for the underlying $\Sigma$-protocol such that
  \begin{align*}\nonumber
  &\Pr_\Theta[V(x,y,\Theta,m)=1:(y,m)\gets \langle\mathcal{S}^\A,\Theta\rangle] \\
  & \geq \frac{1}{2(2q+1)(2q+3)}
  \Pr_\H[V(x,y,\H(x,y),m)=1,~(y,m)\gets\A^\H(x)] - \frac{1}{(2q+1)|\Y|}.
  \end{align*}
\end{theorem}
In the above, $(y,m) \gets \langle\mathcal{S}^\A,\Theta\rangle$ indicates that $y$ and $m$ are the first-round and third-round (respectively) messages of $\mathcal{S}^\A$, when it is given the random challenge $\Theta$ in the second round. 

\subsection{Extension to generalized $\Sigma$-protocols}

In this section, we show that Fiat-Shamir also preserves soundness for a more general family of protocols, which we call ``generalized $\Sigma$-protocols.'' In such a protocol, $\verifier$ can begin the protocol by sending an initial message to $\prover$.

\begin{protocol}[Generalized $\Sigma$-protocol]\label{prot:generalized-sigma}
Select a public function $f\colon\R\times L\to\W$, a finite set $\C$, and a distribution $D$ over $\R$. The protocol begins with $\prover$ and $\verifier$ receiving an input $x$.
\begin{proto}
\item [Round 1.] $\verifier$ samples randomness $r\in\R$ from distribution $D$ and computes message $w=f(r,x)$, which is sent to $\prover$.
\item [Round 2.] $\prover$ sends a message $y$ to $\verifier$.
\item [Round 3.] $\verifier$ responds with a uniformly random classical challenge $c\in\C$.
\item [Round 4.] $\prover$ sends a response $m$ to $\verifier$.
\item [Verdict.] $\verifier$ outputs a bit computed by a Boolean function $V(r,x,y,c,m)$.
\end{proto}
\end{protocol}

Notice that the original  Mahadev protocol~\cite{Mah18} is a \genprot:
the distribution $D$ describes the distribution for the secret key, and $f$ computes the public key.
Similarly, the $k$-fold parallel repetition of our instance-independent protocol is also a \genprot since our trusted setup phase can be seen as a message from the verifier.

\paragraph{Fiat-Shamir for generalized $\Sigma$ protocols.}\label{sec:FS-generalized-sigma}

The FS transformation for \genprots is similar to standard ones:
 in the Verdict stage, $\verifier$ computes $c=\H(x,w,y)$  and accepts if and only if $V(r,x,y,c,m)=1$.

\begin{protocol}[FS-transformed \genprot]\label{prot:fs-generalized-sigma}
Select a public function $f:\R\times L\to\W$, a finite set $\C$, and a distribution $D$ over $\R$. $\prover$ and $\verifier$ receive an input $x$ and are given access to a random oracle $\H$.
\begin{proto}
\item [Round 1.] $\verifier$ samples randomness $r\in\R$ from distribution $D$, and computes message $w=f(r,x)$, which is sent to $\prover$.
\item [Round 2.] $\prover$ sends a message $(y,m)$ to $\verifier$.
\item [Verdict.] $\verifier$ computes $c=\H(x,w,y)$ and then outputs a bit computed by a Boolean function $V(r,x,y,c,m)$.
\end{proto}
\end{protocol}

To show that \genprots remain secure under FS transformation, similarly to the idea for $\Sigma$-protocols in \app{sigma-reduction},
we give a reduction.
Conditioned on any randomness $r$, the prover $\A_r^\H(x):=\A^\H(x,f(r,x))$ is characterized by unitaries $A_{0,f(r,x)},\ldots,A_{q,f(r,x)}$.
The prover $\B$ in the $\Sigma$-protocol runs $\mathcal{S}^{\A_r}$ and outputs its decision.
Given the success probability of $\A$, we establish a lower bound on that of $\B$, formally stated as follows.

\begin{restatable}[Fiat-Shamir transformation for \genprots]{lemma}{lemfs}\label{lem:genprots}
  Suppose that
\begin{align*}
  \Pr_{r,\H}[V(r,x,y,\H(x,f(r,x),y),m)=1:~(y,m)\gets \A^\H(x,f(r,x))] = \epsilon.
\end{align*}
Then
\begin{align*}
  \Pr_{r,\Theta}[V(r,x,y,\Theta,m)=1:~(y,m)\gets \langle \B,\Theta\rangle] \geq \frac{\epsilon}{2(2q+1)(2q+3)}-\frac{1}{(2q+1)|\Y|}.
\end{align*}
\end{restatable}
\begin{proof}
  See \app{sigma-reduction}.
\end{proof}

\noindent
\lem{genprots} immediately gives the following theorem.
\begin{theorem}\label{thm:g-sigma}
  If a language $L$ admits a \genprot with soundness error $s$,
  then after the Fiat-Shamir transformation, the soundness error against provers who make up to $q$ queries to a random oracle is $O(sq^2+q|\Y|^{-1})$.
\end{theorem}
\begin{proof}
  Suppose that there is a prover who succeeds in the transformed protocol with success probability $\epsilon$.
  Then by \lem{genprots}, we may construct a prover who succeeds with probability at least $\frac{\epsilon}{O(q^2)}-O\left(\frac{1}{q|\Y|}\right)$.
  By the soundness guarantee, we have $\frac{\epsilon}{O(q^2)}-O\left(\frac{1}{q|\Y|}\right)\leq s$
  and thus $\epsilon\leq O(q^2s+q|\Y|^{-1})$.
\end{proof}
By \thm{g-sigma}, if both $s$ and $|\Y|^{-1}$ are negligible in security parameter $\lambda$,
the soundness error of the transformed protocols remains negligible against an efficient prover who makes $q=\poly(\lambda)$ queries. \thm{noninteractive-intro} follows directly from \thm{g-sigma}.

\subsection{Non-interactive zero-knowledge for $\QMA$}

We now show that, using the Fiat-Shamir transformation, our three-round protocol proposed in \prot{interactive} can be converted into a non-interactive zero-knowledge argument (with trusted setup) for $\QMA$ in the Quantum Random Oracle model. The resulting protocol is defined exactly as \prot{interactive}, with two modifications: (i.) instead of Round $\verifier_2$, the prover $\prover$ computes the coins $c$ by evaluating the random oracle $\mathcal H$ on the protocol transcript thus far, and (ii.) the NIZK instance $x$ is appropriately redefined using these coins.

We remark that since the setup in this protocol is trusted, it follows from \thm{g-sigma} that the compressed protocol is complete and sound, and therefore we just need to argue about the zero-knowledge property.

\begin{theorem}
  The Fiat-Shamir transformation of \prot{interactive} is zero-knowledge.
\end{theorem}
\begin{proof}
  The simulator $\simulator^{\verifier_2^*}$ can sample the trapdoor keys for NTCF/NTIF functions and private keys for the FHE scheme, enabling simulation of the transcript for every challenge sent by the verifier. In particular, one can run the same proof of \sec{zk-property} with the variant $\simulator^{\H}$ that queries the random oracle $\H$ for the challenges instead of receiving it from a malicious verifier $\verifier^*$.
\end{proof}

\section*{Acknowledgments}

We thank Kai-Min Chung, Andrea Coladangelo, Bill Fefferman, Serge Fehr, Christian Majenz, Christian Schaffner, Umesh Vazirani, and Thomas Vidick for helpful discussions.

AMC and SHH acknowledge support from NSF grant CCF-1813814 and from the U.S. Department of Energy, Office of Science, Office of Advanced Scientific Computing Research, Quantum Testbed Pathfinder program under Award Number DE-SC0019040. GA acknowledges support from NSF grant CCF-1763736. Part of this work was completed while GA was supported by the Dutch Research Council (NWO) through travel grant 040.11.708.
Part of this work was completed while AG was visiting the
Joint Center for Quantum Information and Computer Science, University of Maryland and the Simons Institute for the Theory of Computing.

\bibliographystyle{plain}
\bibliography{refs}

\appendix

\section{Standard cryptographic primitives}
\label{app:standard-primitives}
\subsection{Commitment schemes}

The following definition is taken from \cite{CVZ19} (with modification).
A trapdoor commitment scheme is a tuple of algorithms $\Com=(\gen,\commit,\verify)$, described as follows.
\begin{enumerate}
  \item $\gen(1^\lambda)$, on input the security parameter, outputs a public key $pk$ and a secret key $sk$.\footnote{The secret key is never used.}

  \item $\commit_{pk}(b,s)$, on input the public key $pk$, a bit $b\in\bit$ (to commit to) and a string $s$, outputs the commitment $z$.

  \item $\verify_{pk}(b,z,s)$, on input the public key $pk$, a bit $b\in\bit$, a string $s$, and the commitment $z$, outputs ``accept'' or ``reject.''
\end{enumerate}
The scheme based on LWE \cite{CVZ19} satisfies the following properties.
\begin{enumerate}
  \item Perfectly binding: if $\commit_{pk}(b,s)=\commit_{pk}(b',s')$, then $b=b'$.
  \item Quantum computational concealing: for any quantum adversary $\A$,
    \begin{align*}
      \Pr\left[\A(pk,z)=b \middle|\begin{aligned} (pk,sk) &\gets \gen(1^\lambda) \\ s &\gets\bit^\ell \\ z &\gets\commit_{pk}(b,s)\end{aligned} \right] \leq \frac{1}{2}+\negl(\lambda).
    \end{align*}
\end{enumerate}

\subsection{Fully homomorphic encryption with circuit privacy}

A fully homomorphic encryption scheme $\FHE=(\Gen,\Enc,\Dec,\Eval)$ with malicious circuit privacy \cite{CVZ19,BS19} consists of the following algoithms.
\begin{enumerate}
  \item $\Gen(1^\lambda)$: a probabilistic algorithm that, on input $1^\lambda$, outputs a secret key $sk$ and a public key $pk$.
  \item $\Enc_{pk}(x)$: a probabilistic algorithm that, on input the public key $pk$ and message $x$, outputs a ciphertext.
  \item $\Dec_{sk}(c)$: a deterministic algorithm that, on input the secret key $sk$ and ciphertext $c$, outputs a message $x$.
  \item $\Eval_{pk}(C,c_1,\ldots,c_m)$: a probabilistic algorithm that, on input the public key $pk$, a circuit description $C\in\C_\lambda$, and ciphertexts $c_1,\ldots,c_m$, outputs another ciphertext $c'$.
\end{enumerate}
The scheme satisfies the following properties for any polynomial-sized classical circuits $\{\C_\lambda\}_\lambda$:
\begin{enumerate}
  \item Correctness: for any $\lambda\in\mathbb{N}$, $m\in\bit^*$ and $C\in\C_\lambda$,
    \begin{align*}
      \Pr\left[\Dec_{sk}(c')=C(x)\middle|\begin{aligned} (pk,sk)&\gets\Gen(1^\lambda) \\ c&\gets \Enc_{pk}(m) \\ c'&\gets\Eval_{pk}(C,c)\end{aligned}\right] = 1.
    \end{align*}

  \item Semantic security against quantum adversaries: for every $\lambda\in\mathbb{N}$, there exists a negligible function $\mu$ such that for any pair of messages $m_0,m_1$ of polynomial size and any quantum adversary $\A$,
    \begin{align*}
      \Pr\left[\A(pk,c)=b\middle|\begin{aligned}b&\gets\bit \\ (pk,sk)&\gets\Gen(1^\lambda) \\ c&\gets\Enc_{pk}(m_b)\end{aligned}\right]
      \leq \frac{1}{2}+\mu(\lambda).
    \end{align*}

  \item Malicious\footnote{In contrast to the semi-honest counterpart, it is not required that the public key and the ciphertext to $\Eval$ are well-formed.} circuit privacy: there exist unbounded algorithms $\Sim$, $\Ext$ such that for any $x\in\bit^*$, possibly malformed $pk^*$, and
    $ct\gets\Enc_{pk^*}(x)$, we have $\Ext_{pk^*}(1^\lambda,ct)=x$.
    Furthermore, for any $C$ and possibly malformed $pk^*,ct^*$,
    \begin{align*}
      \Eval_{pk^*}(C,ct^*) \approx_s \Sim_{pk^*}(C(\Ext_{pk^*}(ct^*;1^\lambda));1^\lambda),
    \end{align*}
    where $\approx_s$ denotes that the two distributions are statistically indistinguishable.
\end{enumerate}
An FHE scheme with malicious circuit privacy \cite{OPCPC14} is known to exist, assuming that LWE is (quantum) computationally intractable.

\section{Proof of \lem{modified-mf}}\label{app:modifiedmf}

\modifiedmf*
\begin{proof}
First we observe that for each copy, with probability $1/4$, $\verifier$ measures the quantum state with a term sampled from the distribution $\pi$; otherwise $\verifier$ accepts.
Thus for an instance $H$, the effective Hamiltonian to verify is $\widetilde{H}^{\otimes r}$ where $\widetilde{H}=\frac{3\1+H}{4}$.
Following the standard parallel repetition theorem for $\QMA$, we know that $\prover$'s optimal strategy is to present the the ground state of $\widetilde{H}$, which is also the ground state of $H$.

With probability $\binom{r}{t}(\frac{1}{4})^t(\frac{3}{4})^{r-t}$, there are $t$ consistent copies.
Now for $i\in A$, we let $X_i$ be a binary random variable corresponding to the decision of $\verifier_i$.
For soundness,
by Hoeffding's inequality\footnote{$\Pr[\frac{1}{n}\sum_i X_i-\mu \geq\delta]\leq e^{-2t\delta^2}$ for i.i.d.\ $X_1,\ldots,X_n\in[0,1]$.}
the success probability for $A$ such that $|A|=t$ is
\begin{align*}\nonumber
  \Pr[\text{accept}|A]
  &= \Pr\left[\frac{1}{t}\sum_{i\in A} X_i \geq \frac{c+s}{2} \right]  \\
  &\leq \Pr\left[\frac{1}{t}\sum_{i\in A} X_i - s \geq \frac{c-s}{2} \right]
  \leq 2e^{-\frac{tg^2}{2}},
\end{align*}
where $g=c-s$ is the promise gap.
Then the overall success probability is
\begin{align}\label{eq:soundness}\nonumber
  \Pr[\text{accept}]
  &= 2\cdot 4^{-r}\sum_{t=0}^r\binom{r}{t}3^{r-t} e^{-tg^2/2} \\
  &= 2\left(\frac{e^{-g^2/2}+3}{4}\right)^r
  \leq 2(1-g^2/16)^r \leq 2e^{-rg^2/16}
\end{align}
since $1-x/2\geq e^{-x}$ for $x\in[0,1]$ and $1-x\leq e^{-x}$ for $x\geq 0$.
Thus $r=\omega(g^{-2}\log n)$ suffices to suppress the soundness error to $n^{-\omega(1)}$.
Since $g^{-1}=\poly(n)$, polynomially many copies suffice to achieve negligible soundness error.

For completeness, again by Hoeffding's inequality,
\begin{align*}\nonumber
  \Pr[\text{reject}|A]
  &= \Pr\left[\frac{1}{t}\sum_{i\in A} X_i < \frac{c+s}{2} \right] \\
  &\leq \Pr\left[c-\frac{1}{t}\sum_{i\in A} X_i  > \frac{c-s}{2} \right]
  \leq 2e^{-\frac{tg^2}{2}}.
\end{align*}
By the same calculation as in \eq{soundness}, the completeness error is negligible if we set $r=\omega(g^{-2}\log n)$.
\end{proof}

\section{Proof of \lem{nearly-orthogonal-measurements}}\label{app:lemma-orthogonal-measurements}
\ortm*
\begin{proof}
Since we are proving an upper bound for a quantity that is symmetric under the interchange of $b$ and $a$, we can assume that $b_i=0$ and $a_i=1$ without loss of generality.

We first claim that there exists a quantum state $\rho$ such that
\begin{equation}\label{eq:desired}
\Exp_{(pk,sk)\gets\Gen(1^\lambda,h)}\left[\bra{\Psi_{pk}}\Pi_{s,sk,b}^{U_b}\Pi_{s,sk,a}^{U_a}\Pi_{s,sk,b}^{U_b}\ket{\Psi_{pk}}\right]
\leq \alpha_{h_i,s_i,\rho}+\negl(n)
\end{equation}
for all basis choices $h$ and randomness $s$. For a contradiction, suppose that is not the case.
Then there exists a basis choice $h^*$ and $s^*$ and a polynomial $\eta$ such that for every state $\rho$,
\begin{align*}
\Exp_{(pk,sk)\gets\Gen(1^\lambda,h^*)}\left[\bra{\Psi_{pk}}\Pi_{s^*,sk,b}^{U_b}\Pi_{s^*,sk,a}^{U_{a}}\Pi_{s^*,sk,b}^{U_b}\ket{\Psi_{pk}}\right]
>\alpha_{h_i^*,s_i^*,\rho}+1/\eta(n)\,.
\end{align*}
We show that this implies the existence of an efficient prover $\prover^*$ for the single-copy three-round Mahadev protocol $\mathfrak{M}$ who violates \lem{perfect-succ-prob}.
Define the following projector on $WXYE$:
  \begin{align*}
    \Sigma_a := U_a^\dag(H^a\otimes\1_E)((\1\otimes\cdots\otimes\1\otimes\Pi \otimes\1\otimes\cdots\otimes\1)\otimes\1_E)(H^a\otimes\1_E)U_a\,.
  \end{align*}
Here $\Pi$ denotes the single-copy protocol acceptance projector for the Hadamard round, with key $sk_i$ and basis choice $h^*_i,s_i^*$.
In the above, $\Pi$ acts on the $i$th set of registers, i.e., $W_iX_iY_i$. The projector $\Sigma_a$ corresponds to performing the appropriate Hadamard test in the $i$th protocol copy, and simply accepting all other copies unconditionally. It follows that $\Pi_{s,sk,a}^{U_{a}}\preceq \Sigma_a$, and we thus have
\begin{align}\label{eq:contra}\nonumber
    &\Exp_{(pk,sk)\gets\Gen(1^\lambda,h^*)}\left[\bra{\Psi_{pk}}\Pi_{s^*,sk,b}^{U_b}\Sigma_a\Pi_{s^*,sk,b}^{U_b}\ket{\Psi_{pk}}\right]\\\nonumber
    &\qquad\geq \Exp_{(pk,sk)\gets\Gen(1^\lambda,h^*)}\left[\bra{\Psi_{pk}}\Pi_{s^*,sk,b}^{U_b}\Pi_{s^*,sk,a}^{U_{a}}\Pi_{s^*,sk,b}^{U_b}\ket{\Psi_{pk}}\right] \\
    &\qquad>\alpha_{h^*_i,s_i^*,\rho}+1/\eta.
\end{align}
The single-copy prover $\prover^*$ interacts with the single-copy verifier $\verifier^*$ as follows.
\begin{itemize}
\item In the Setup phase, after receiving the public key $pk^*$, initialize $k-1$ internally simulated verifiers, and set $pk$ to be the list of their keys, with $pk^*$ inserted in the $i$th position. Let $h = (h_1, \dots, h_k)$ be the basis choices, and note that all but $h_i$ are known to $\prover^*$.
\item Using the algorithms of $\prover$, perform the following repeat-until-success (RUS) procedure for at most $q=\eta^4$ steps.
\begin{enumerate}
\item Prepare the state $\ket{\Psi_{pk}}$ on registers $WXYE$, and then apply the unitary $U_b$.
\item Apply the measurement determined by $\Pi_{s, sk, b}$ (defined in \eq{prod-meas}); for index $i$ we can use $pk^*$ because $b_i = 0$; for the rest we know the secret keys.
\item If the measurement rejects, go to step (1.), and otherwise apply $U_b^\dagger$ and output the state.
\end{enumerate}
If the RUS procedure does not terminate within $q$ steps, then $\prover^*$ prepares a state\footnote{To pass the test round, any efficiently preparable state suffices.} $\ket{\Phi_{pk}^*}$ by performing $\Samp$ coherently on $\ket{0^n}_W$ (see Round~2 of \prot{Mahadev}).

Note that if $\prover^*$ terminates within $q$ steps, the resulting state is
\begin{align*}
\ket{\Phi_{pk}} : = \frac{\Pi_{s^*,sk,b}^{U_b}\ket{\Psi_{pk}}}{\|\Pi_{s^*,sk,b}^{U_b}\ket{\Psi_{pk}}\|}\,;
\end{align*}
otherwise $\ket{\Phi_{pk}^*}$ is prepared.
\item For the Round~$\prover_1$ message, measure the $Y_i$ register of $\ket{\Phi_{pk}}$ and send the result to $\verifier^*$.
\item When $\verifier^*$ returns the challenge bit $w$ in Round 3, if $w = b_i = 0$, apply $U_b$ (resp. $\1$) to $\ket{\Phi_{pk}}$ (resp. $\ket{\Phi_{pk}^*}$), and otherwise apply $U_a$. Then behave honestly, i.e., measure $W_iX_i$ in computational or Hadamard bases as determined by $w$, and send the outcomes.
\end{itemize}
By the RUS construction and the fact that $b_i = 0$, the state $\ket{\Phi_{pk}}$ or $\ket{\Phi_{pk}^*}$ is in the image of the test-round acceptance projector in the $i$th coordinate. This means that, when $\verifier^*$ enters a test round, i.e., $w = 0 = b_i$, $\prover^*$ is accepted perfectly. In other words, $\prover^*$ is a perfect prover\footnote{While we used $\Pi_{h^*,sk,b}$ in the RUS procedure, and $h_i^*$ is (almost always) not equal to the $h_i$ selected by $\verifier^*$, the result is still a perfect prover state. This is because, as described in \prot{Mahadev}, the acceptance test in the test round is independent of the basis choice.} and thus satisfies the hypotheses of \lem{perfect-succ-prob}.

Now consider the case when $\verifier^*$ enters a Hadamard round, i.e., $w=1$.
Let
\begin{align*}
  \Omega:=\{(pk,sk):\bra{\Psi_{pk}}\Pi_{s^*,sk,b}^{U_b}\ket{\Psi_{pk}}>q^{-1/2}\}
\end{align*}
denote the set of ``good'' keys.
For $(pk,sk)\in\Omega$, the probability of not terminating within $q = \poly(n)$ steps is at most $(1-q^{-1/2})^q \le e^{-\sqrt q}$.  Therefore, the success probability of RUS for the good keys is $1-\negl(n)$. Thus we have
\begin{align*}
\Exp_{sk|\Omega}[\bra{\Phi_{pk}}\Sigma_a\ket{\Phi_{pk}}]\Pr[\Omega]
\leq\alpha_{h_i^*,s_i^*,\rho}+\negl(n)
\end{align*}
where we let $\Exp_{X|E}[f(X)]:=\frac{1}{\Pr[E]}\sum_{x\in E}p(x)f(x)$ denote the expectation value of $f(X)$ conditioned on event $E$ for random variable $X$ over finite set $\X$ with distribution $p$ and function $f\colon\X\to[0,1]$.
Now we divide \eq{contra} into two terms and find
\begin{align*}\nonumber
  \alpha_{h_i^*,s_i^*,\rho}+\eta^{-1}
  &< \Exp_{(pk,sk)}\left[\bra{\Psi_{pk}}\Pi_{s^*,sk,b}^{U_b}\Sigma_a\Pi_{s^*,sk,b}^{U_b}\ket{\Psi_{pk}}\right] \\\nonumber
  &=
    \Pr[\Omega]
    \Exp_{(pk,sk)|\Omega}\left[\bra{\Psi_{pk}}\Pi_{s^*,sk,b}^{U_b}\Sigma_a\Pi_{s^*,sk,b}^{U_b}\ket{\Psi_{pk}}\right] \\\nonumber
  &\qquad + \Pr[\overline{\Omega}]
    \Exp_{(pk,sk)|\overline{\Omega}}\left[\bra{\Psi_{pk}}\Pi_{s^*,sk,b}^{U_b}\Sigma_a\Pi_{s^*,sk,b}^{U_b}\ket{\Psi_{pk}}\right] \\\nonumber
  & \leq
    \Pr[\Omega]
    \Exp_{(pk,sk)|\Omega}\left[\bra{\Psi_{pk}}\Pi_{s^*,sk,b}^{U_b}\Sigma_a\Pi_{s^*,sk,b}^{U_b}\ket{\Psi_{pk}}\right]
    + q^{-1/2} \\
  & \leq
    \alpha_{h_i^*,\rho} + \negl(n) + q^{-1/2}.
\end{align*}
Since $q=\eta^4$, this is a contradiction. Therefore \eq{desired} holds for every $h,s$, i.e.,
$$
\Exp_{(pk,sk)\gets\Gen(1^\lambda,h)}[\bra{\Psi_{pk}}\Pi_{s,sk,b}^{U_b}\Pi_{s,sk,a}^{U_a}\Pi_{s,sk,b}^{U_b}\ket{\Psi_{pk}}]\leq \alpha_{h_i,s_i,\rho}+\negl(n).
$$
It then follows that
  \begin{align*}\nonumber
    &\Exp_{(pk,sk)\gets\Gen(1^\lambda,h)}\left[\bra{\Psi_{pk}}\Pi_{h,sk,b}^{U_b}\Pi_{h,sk,a}^{U_{a}}+\Pi_{h,sk,a}^{U_{a}}\Pi_{h,sk,b}^{U_{b}}\ket{\Psi_{pk}}\right] \\\nonumber
    &\qquad =2\Exp_{(pk,sk)\gets\Gen(1^\lambda,h)}\left[\Re(\bra{\Psi_{pk}}\Pi_{h,sk,b}^{U_b}\Pi_{h,sk,a}^{U_{a}}\ket{\Psi_{pk}})\right] \\\nonumber
    &\qquad \leq 2\Exp_{(pk,sk)\gets\Gen(1^\lambda,h)}\left[|\bra{\Psi_{pk}}\Pi_{h,sk,b}^{U_b}\Pi_{h,sk,a}^{U_{a}}\ket{\Psi_{pk}}|\right] \\\nonumber
    &\qquad \leq 2\Exp_{(pk,sk)\gets\Gen(1^\lambda,h)}\left[\bra{\Psi_{pk}}\Pi_{h,sk,b}^{U_b}\Pi_{h,sk,a}^{U_{a}}\Pi_{h,sk,b}^{U_b}\ket{\Psi_{pk}}^{1/2}\right] \\
    &\qquad \leq 2\Exp_{(pk,sk)\gets\Gen(1^\lambda,h)}\left[\bra{\Psi_{pk}}\Pi_{h,sk,b}^{U_b}\Pi_{h,sk,a}^{U_{a}}\Pi_{h,sk,b}^{U_b}\ket{\Psi_{pk}}\right]^{1/2}
      \leq 2\alpha_{h_i,s_i,\rho}^{1/2}+\negl(n)
  \end{align*}
as claimed.
\end{proof}

\section{Proof of \lem{orthogonality-implies-soundness}}\label{app:lemmas}

\ort*
\begin{proof}
  Let $\alpha:=\bra{\psi}A_1+\ldots+A_m\ket{\psi}$. We have
  \begin{align}
    \alpha^2
    &\leq \bra{\psi}(A_1+\cdots+A_m)^2\ket{\psi} \nonumber \\
    &= \alpha + \sum_{i< j}\bra{\psi}A_iA_j+A_jA_i\ket{\psi} \label{eq:idempotent} \\
    &\leq \alpha + \sum_{i<j}\delta_{ij} \nonumber
  \end{align}
  The first inequality  holds since $\proj{\psi}\preceq\1$, and thus
  \begin{align*}
    \bra{\psi}(A_1+\cdots+A_m)\proj{\psi}(A_1+\cdots+A_m)\ket{\psi}\leq \bra{\psi} (A_1+\cdots+A_m)^2\ket{\psi}.
  \end{align*}
  The equality \eq{idempotent} holds since each $A_i$ is idempotent, and thus
  \begin{align*}
    \bra{\psi}(A_1+\cdots+A_m)^2\ket{\psi}
    &= \bra{\psi}A_1^2+\cdots+A_m^2\ket{\psi} + \sum_{i<j}\bra{\psi}A_iA_j+A_jA_i\ket{\psi} \\
    &= \bra{\psi}A_1+\cdots+A_m\ket{\psi} + \sum_{i<j}\bra{\psi}A_iA_j+A_jA_i\ket{\psi}.
  \end{align*}
  Now observe that for $\beta>0$, $x^2\leq x+\beta$ implies $x\leq \frac{1}{2}(1+\sqrt{1+4\beta})\leq \frac{1}{2}(1+(1+2\sqrt{\beta}))=1+\sqrt{\beta}$. Thus
  $\alpha\leq 1+\sqrt{\sum_{i<j}\delta_{ij}}$ as claimed.
\end{proof}

\section{Proof of soundness for \prot{interactive}}
\label{app:soundness-zk}

Recall that in \thm{soundness-zk}, we define the following hybrids:
\hybrids

Let $A_i$ be the maximum acceptance probability of the prover in the hybrid $H_i$ among all no-instances. We show that $A_i \leq \negl(n)$ for all $i \in \{0,\ldots,4\}$. Notice that $A_0 \leq \negl(n)$ by \thm{interactive}.

\begin{lemma}\label{lem:hyb1}
  $A_1 \leq \negl(n)$.
\end{lemma}
\begin{proof}
  Let the soundness error of $H_0$ be $\epsilon=\negl(n)$.
  First we introduce a hybrid protocol $H_{0.1}$, which is the same as $H_1$ except that in the setup phase, both $\prover$ and $\verifier$ receive a uniform key pair $(\beta,\gamma)$ (and the instance $H$), and $\verifier$ verifies the instance $H_{\beta,\gamma}$, which has the same ground state energy as $H$.

  Conditioned on any key pair $(\beta,\gamma)$, let the success probability of any prover $\prover$ be $\epsilon_{\beta,\gamma}$ in $H_{0.1}$.
  Furthermore, since we have fixed a pair $(\beta,\gamma)$, $H_{0.1}$ is exactly the same as $H_0$ except that $H$ is replaced with $H_{\beta,\gamma}$.
  Thus we conclude that $\epsilon_{\beta,\gamma}\leq\epsilon$.
  The success probability of $\prover$ is $\Exp_{\beta,\gamma}[\epsilon_{\beta,\gamma}]\leq\epsilon$, and thus $H_{0.1}$ has soundness error $\epsilon$.

  Next, we observe that $H_1$ is the same as $H_{0.1}$ except that $\verifier$ receives a commitment $\xi$ to a uniform key pair $(\beta,\gamma)$ in the setup phase, and $\prover$ reveals the key pair in Round $\prover_2$.
  Since $\verifier$'s only message (Round~$\verifier_2$) does not depend on the key pair in either $H_{0.1}$ or $H_1$ (public coins) and the commitment is perfectly binding (and therefore the prover cannot change the values of $\beta$ and $\gamma$ depending on the verifier's challenge), in the reduction, $\verifier$ does not reject with higher probability if $\verifier$ knows the keys in an earlier step.
\end{proof}

\begin{lemma}\label{lem:hyb2}
  $A_2 \leq  \negl(n)$.
\end{lemma}
\begin{proof}
  The difference between $H_1$ and $H_2$ is whether $\prover$ sends a commitment to $u$ in Round $\prover_2$, and $\verifier$ checks if $\chi$ is a commitment to $u$.
  Since the probability can only become smaller under this change, the success probability in $H_2$ is negligible.
\end{proof}

\begin{lemma}\label{lem:hyb3}
  $A_3 \leq \negl(n)$.
\end{lemma}
\begin{proof}
  Let us suppose that there is a prover who wins $H_{3}$ with non-negligible probability $\epsilon$. We can construct an adversary in $H_2$ that simulates such a prover by providing encryptions of $(0,0)$ instead of $(s,sk)$. By the security of the FHE scheme, such an adversary cannot distinguish both ciphertexts, and therefore $H_3$ would still make the verifier accept with probability at least $\epsilon - \negl(n)$. In this case the adversary in $H_2$ would succeed with the same probability, which is a contradiction.
\end{proof}
\begin{lemma}\label{lem:hyb4}
  $A_4 \leq \negl(n)$.
\end{lemma}
\begin{proof}
  In $H_3$, by the perfect binding property of the commitment scheme, for an accepting view $x=(H,s,sk,\xi,y,c,\chi)$,
  there is a unique witness $\witness=(u,\beta,\gamma,r_1,r_2)$ such that $\verdict'(x,\witness)=1$; otherwise $x$ is unsatisfiable.
Since no efficient quantum prover can generate $(\xi,y,\chi)$ such that $x \in \LL$ except with negligible probability (otherwise it would be possible for a prover in $H_3$ to succeed with non-negligible probability),
for any prover $\prover$ and $H\in\zx_{no}$,
  \begin{align}\label{eq:no}
    \Pr_{s,sk,c}[(\xi,y,\chi)\gets \prover:~(H,s,sk,\xi,y,c,\chi)\notin\LL]>1-\delta
  \end{align}
  for some negligible function $\delta$.

  By the soundness property of the $\NIZK$ scheme, no bounded provers could then provide a proof that makes the verifier accept such a no instance except with negligible probability. The fact that the $\NIZK$ is encrypted does not change this.
\end{proof}

\section{The reduction for $\Sigma$-protocols in the QROM}\label{app:QROM}
\label{app:sigma-reduction}

We now briefly describe the reduction in \expref{Theorem}{thm:FS-sigma}.

While this is nontrivial to prove, the Fiat-Shamir transformation remains secure in the QROM, up to some mild conditions on the $\Sigma$-protocol~\cite{DFMS19}. The proof proceeds by a reduction that uses a successfully cheating prover $\A^\H$ in the FS-transformed protocol to build a successfully cheating prover $\mathcal S$ in the $\Sigma$-protocol. This can actually be done in a black-box way, so that $\mathcal S$ only needs to query $\A$.

We let $Y,C,M,E$ denote the registers storing the commitment, challenge, response, and prover workspace, respectively. A $q$-query prover $\A^\H$ is characterized by a sequence of unitaries $A_0,A_1,\ldots,A_q$ on registers $Y,C,M,E$. After $i$ queries, it is in the state
\begin{align*}
  \ket{\psi_i^\H} := A_i U_\H A_{i-1} U_\H \ldots A_1 U_\H A_0 \ket{0}_{YCME}\,.
\end{align*}
In a normal execution, $\A^\H$ first prepares $\ket{\psi_q^\H}$, then measures $Y$ and $M$ in the standard basis and sends the outcome to $\verifier_{FS}$ (see \prot{FS-sigma}). The success probability is
\begin{align*}
  \Pr_\H[V(x,y,\H(x,y),m)=1,~(y,m)\gets\A^\H(x)].
\end{align*}

The prover $\mathcal{S}^\A$ begins by internally instantiating a quantum-secure pseudorandom function\footnote{If an upper bound on $q$ is known, a $2q$-wise independent function would suffice.} $\F$~\cite{Zha12}. It then chooses a random index $i\in\{0,\ldots,q\}$ and prepares the state $\ket{\psi_i^{\F}}$. Note that the queries of $\A$ are answered using $\F$. Next, the simulator performs a standard basis measurement on register $Y$ and sends the outcome $y$ to $\verifier$. After $\verifier$ returns a random challenge $\Theta$, $\mathcal{S}^\A$ constructs a reprogrammed oracle, denoted by $\F*\Theta y$, with
\begin{align*}
  \F * \Theta y (x, y') = \begin{cases}\Theta & \text{if } y'= y \\ \F(x,y') & \text{otherwise}\end{cases}
\end{align*}
which will be used for the remaining simulation of $\A$.
$\mathcal{S}^\A$ then tosses a random coin $b$ and performs the following:
\begin{enumerate}
  \item If $b=0$, $\mathcal{S}^\A$ runs $\A$ with the reprogrammed oracle $\F*\Theta y$ for query $i+1,\ldots,q$.
  \item If $b=1$, $\mathcal{S}^\A$ runs $\A$ with the original oracle $\F$ for query $i+1$ and with $\F*\Theta y$ for query $i+2,\ldots,q$.
\end{enumerate}
To generate the message to $\verifier$, $\mathcal{S}^\A$ measures registers $Y,M$ and obtains the outcomes $y'',m$.
If $y''\neq y$ then $\mathcal{S}^\A$ aborts; otherwise $\mathcal{S}^\A$ outputs the measurement outcome $(y,m)$.
The success probability of $\mathcal{S}^\A$ is
\begin{align*}
  \Pr_\Theta[V(x,y,\Theta,m)=1:(y,m)\gets \langle\mathcal{S}^\A,\Theta\rangle],
\end{align*}
where $\langle\mathcal{S}^\A,\Theta\rangle$ denotes the interaction between $\verifier$, who sends a challenge $\Theta$, and $\mathcal{S}^\A$ in the $\Sigma$-protocol.
\expref{Theorem}{thm:FS-sigma} by \cite{DFMS19}  establishes that the success probabilities of $\mathcal{S}^\A$ and $\A^\H$ are polynomially related.

Now we use the same reduction to prove \lem{genprots}.
\lemfs*
\begin{proof}
Conditioned on $r=r^*$, the success probability of $\A$ is
\begin{align*}\nonumber
  &\Pr_\H[V(r,x,y,\H(x,f(r,x),y),m)=1:~(y,m)\gets \A^\H(x,f(x,r))|r=r^*] \\
  &\qquad=\Pr_{\H_{r^*}}[V(r^*,x,y,\H_{r^*}(x,y),m)=1:~(y,m)\gets \A_{r^*}^\H(x)] =:
  \epsilon(r^*),
\end{align*}
where $\H_{r^*}(x,y):=\H(x,f(r,x),y)$
and by definition $\Exp_r[\epsilon(r)]=\epsilon$.
Since $\H$ is a random function over $L\times\W\times\Y\to\C$, for any $r^*$, $\H_{r^*}$ is a random function over $L\times\Y\to\C$.
Then by \thm{FS-sigma}, the success probability of $\mathcal{S}^{\A_{r^*}}$ is
\begin{align*}
  \Pr_\Theta[V(r^*,x,y,\Theta,m)=1:~(y,m)\gets \langle \mathcal{S}^{\A_{r^*}},\Theta\rangle]
  \geq
  \frac{\epsilon(r^*)}{2(2q+1)(2q+3)}-\frac{1}{(2q+1)|\Y|}.
\end{align*}
Taking the average over $r^*$, we have
\begin{align*}
  &\Pr_{r,\Theta}[V(r,x,y,\Theta,m)=1:~(y,m)\gets \langle\B,\Theta\rangle] \\\nonumber
  & \qquad = \Exp_{r^*}\left[\Pr_\Theta[V(r^*,x,y,\Theta,m)=1:~(y,m)\gets \langle\mathcal{S}^{\A_{r^*}},\Theta\rangle|r=r^*]\right] \\
  & \qquad = \Exp_{r^*}\left[
    \frac{\epsilon(r^*)}{2(2q+1)(2q+3)}-\frac{1}{(2q+1)|\Y|}
    \right] \\
  & \qquad = \frac{\epsilon}{2(2q+1)(2q+3)}-\frac{1}{(2q+1)|\Y|}
\end{align*}
as claimed.
\end{proof}

\end{document}